\documentclass[twoside,leqno]{article}
\usepackage[utf8]{inputenc}

\usepackage[letterpaper]{geometry}

\usepackage{ltexpprt}
\usepackage{hyperref}

\usepackage{listings}
\usepackage[algo2e,algoruled]{algorithm2e}
\usepackage{amsfonts}       %
\usepackage{amsmath}
\usepackage{amssymb}

\usepackage{bbm}
\usepackage{bm}
\usepackage{color}
\usepackage{booktabs}       %
\usepackage{caption}
\usepackage[T1]{fontenc}    %
\usepackage{graphicx}
\usepackage{hyperref}       %
\usepackage[utf8]{inputenc} %
\usepackage{mathtools}
\usepackage{microtype}      %
\usepackage{multirow}
\usepackage{newtxmath}      %
\usepackage{nicefrac}       %
\usepackage{titlesec}
\usepackage{thm-restate}
\usepackage{enumitem}
\usepackage{tikz}
\usepackage{url}            %
\usepackage{subcaption}     %
\usepackage{wrapfig}        %
\usepackage{caption}
\DeclareMathOperator*{\argmin}{arg\,min}

\newcommand{\ceil}[1]{\lceil #1 \rceil}
\newcommand{\floor}[1]{\lfloor #1 \rfloor}
\newcommand{\eps}{\epsilon}

\newcommand{\bbR}{\mathbb{R}}
\newcommand{\R}{\bbR}
\newcommand{\N}{\mathbb{N}}
\newcommand{\Ind}{\vvmathbb 1}

\let\1\Ind

\newcommand{\position}{m\xspace}

\newcommand{\ignore}[1]{}

\newcommand{\cost}{\mathrm{cost}}

\newcommand{\on}{\mathrm{On}}
\newcommand{\Off}{\mathrm{Off}}
\newcommand{\OPT}{\mathrm{OPT}}

\newcommand{\FIF}{\textsc{FiF}\xspace}
\newcommand{\Blind}{\textsc{BelPred}\xspace}

\newcommand{\tM}{\widetilde M}

\newcommand{\pos}{\mathrm{pos}}

\newcommand{\entpot}{\mathcal{E}}
\newcommand{\reppot}{\mathcal{R}}
\newcommand{\scatpot}{\mathcal{S}}
\newcommand{\overpot}{\Phi}

\title{\Large Learning-Augmented Weighted Paging}
\author{Nikhil Bansal\thanks{University of Michigan, Ann Arbor. \texttt{bansal@gmail.com}. The research was done while the author was at CWI and it was supported  by the NWO VICI grant 639.023.812.}
\and Christian Coester\thanks{Tel Aviv University, Israel. \texttt{christian.coester@gmail.com}. Christian Coester is supported by the Israel Academy of Sciences and Humanities \& Council for Higher Education Excellence Fellowship Program for International Postdoctoral Researchers. Research was carried out while he was at CWI in Amsterdam and supported by the NWO VICI grant 639.023.812.} \and Ravi Kumar\thanks{Google Research, Mountain View. \texttt{ravi.k53@gmail.com, mpurohit@google.com, erikvee@google.com}} \and Manish Purohit\footnotemark[3] \and Erik Vee\footnotemark[3]}
\date{}

\begin{document}

\maketitle

\begin{abstract}
\small\baselineskip=9pt
We consider a natural semi-online model for weighted paging, where at any time the algorithm is given predictions, possibly with errors, about the next arrival of each page.
The model is inspired by Belady's classic optimal offline algorithm for unweighted paging, and extends the recently studied model for learning-augmented paging~\cite{LV18, Rohatgi19, Wei20} to the weighted setting.

\smallskip

For the case of perfect predictions,
we  provide an $\ell$-competitive deterministic and an $O(\log \ell)$-competitive randomized algorithm, where $\ell$ is the number of distinct weight classes. Both these bounds are tight, and imply an $O(\log W)$- and $O(\log \log W)$-competitive ratio, respectively, when the page weights lie between $1$ and $W$. Previously, it was not known how to use these predictions in the weighted setting and only bounds of $k$ and $O(\log k)$ were known, where $k$ is the cache size.
Our results also generalize to the interleaved paging setting and to the case of imperfect predictions, with the %
competitive ratios degrading smoothly from $O(\ell)$ and $O(\log \ell)$ to $O(k)$  and $O(\log k)$, respectively, as the prediction error increases.

\smallskip

Our results are based on several insights on structural properties of Belady's algorithm and the sequence of page arrival predictions, and novel potential functions that incorporate these predictions.
For the case of unweighted paging, the results imply a very simple potential function based proof of the optimality of Belady's algorithm, which may be of independent interest.

\end{abstract}
\section{Introduction}
\label{sec:intro}
Paging is among the most classical and well-studied problems in online computation.  
Here, we are given a universe of $n$ pages and a cache that can hold up to $k$ pages.  At each time step, some page is requested, and if it is not in the cache (called a \emph{cache miss} or \emph{page fault}), it must be fetched into the cache (possibly evicting some other page), incurring a unit cost. The goal of the algorithm is to minimize the total cost incurred. 
The problem is well understood through the lens of competitive analysis~\cite{sleator1985amortized}, with several optimal $k$-competitive deterministic and $O(\log k)$-competitive randomized algorithms known for it \cite{achlioptas2000competitive,fiat1991competitive,mcgeoch1991strongly}.
A remarkable property of paging is that the offline optimum can be computed with rather limited knowledge of the future: only the \emph{relative order} of the \emph{next} request times for pages. In particular, Belady's classic Farthest in Future (\FIF) algorithm~\cite{belady}, which at any time greedily evicts the page whose next request is farthest in the future, gives the optimal solution.

A natural and well-studied generalization of paging is {\em weighted} paging, where each page $p$ has an arbitrary fetching cost $w_p>0$, and the goal is to minimize the total cost.
Besides the practical motivation, weighted paging is very interesting theoretically as the phase-based analyses for unweighted paging do not work anymore (even if there are only two different weights), %
and as it is a stepping stone in the study of more general problems such as metrical task systems (MTS)~\cite{BLS} and the $k$-server problem\footnote{Weighted paging is equivalent to the $k$-server problem on weighted star metrics.}~\cite{MMS}. In fact, $O(\log k)$-competitive randomized algorithms for weighted paging were obtained relatively recently, and required new techniques such as the primal-dual method \cite{BansalBN12,bansal2012primal} and entropic regularization \cite{BubeckCLLM18}. These ideas have been useful for various other problems and also for MTS and the $k$-server problem \cite{bansal2015polylogarithmic, BubeckCLLM18, BubeckCLL19,CoesterL19}.

\paragraph{Learning-augmented setting.} 
Motivated by advances in machine learning, Lykouris and Vassilvitskii  \cite{LV18} recently introduced a %
new semi-online model where at each step, the algorithm has access to some, possibly erroneous, machine-learned advice about future requests and studied the paging problem in this model.
Here, at each time $t$, along with the current page request we are also given the predicted arrival time for the next request of the same page. 
This can be viewed as generalizing the setting for Belady's \FIF algorithm to allow incorrect predictions. 
They design an algorithm with competitive ratio $O(1)$ when the predictions are accurate, and which degrades smoothly as the prediction error increases, but never exceeds $O(\log k)$.
These results have been subsequently refined and improved in \cite{Rohatgi19, Wei20}.

In this work, we study whether Belady's algorithm and the results in the learning-augmented setting for unweighted paging can be extended to the weighted case.
Suppose each page weight is one of distinct values $w_1, \ldots, w_\ell$; the pages are thus divided into $\ell$ disjoint \emph{weight classes}.
Then recent work by Jiang et al.~\cite{Jiang2020online} and Antoniadis et al.~\cite{AntoniadisCEPS20} shows that even with perfect predictions, any deterministic (resp.,~randomized) online algorithm must have competitive ratio $\Omega(\ell)$ (resp.,~$\Omega(\log \ell)$), provided $\ell \leq k$.\footnote{The lower bounds in~\cite{Jiang2020online,AntoniadisCEPS20} are stated in terms of $k$ rather than $\ell$, using a construction with $\ell\approx k$. The effective cache size can be reduced to $\ell$ by forcing $k-\ell$ pages to be in cache at all times.} In particular, for $\ell \geq k$,  predictions do not give any advantage.

As Belady's algorithm is $1$-competitive for $\ell=1$, this raises the natural question whether there are algorithms with guarantees that are only a function of $\ell$, and \emph{independent} of the cache size $k$.
In typical scenarios $\ell$ is likely to be small and much less than $k$. Also if the weights range from $1$ to $W$, then one can assume $\ell = O(\log W)$ by rounding them to powers of $2$. %

\subsection{Prediction model and error}
\label{sec:model-and-error}
We consider the following model for learning-augmented weighted paging. %
At each time $t=1,\dots,T$, the algorithm receives a request to some page $\sigma_t$ as well as a prediction $\tau_t\in\N$ for the next time after $t$ when $\sigma_t$ will be requested again. Let $a_t\in\N$ be the actual time when $\sigma_t$ is next requested (or $a_t=T+1$ if it is not requested again). %
In the unweighted setting of \cite{LV18,Rohatgi19,Wei20}, the prediction error was defined as the $\ell_1$-distance between $a$ and $\tau$, which in the weighted case generalizes naturally to %
\[ \eta := \sum_t w_{\sigma_t} \cdot |\tau_t-a_t|.\]
We remark that although the predictions $\tau_t$ are for the arrival times, we use them only to get a relative ordering of pages within the same weight class by their next predicted arrival times.%

We define the following more nuanced error measure that allows us to obtain tighter bounds. For any weight class $i$, we call a pair $(s,t)$ of time steps an \emph{inversion} if both $\sigma_s$ and $\sigma_t$ belong to weight class $i$ and $a_{s}<a_{t}$ but $\tau_{s}\ge \tau_t$. Let $\epsilon_i(\sigma, \tau) := | \{s \in \N \mid \exists t \in \N\colon \text{ $(s,t)$ is an inversion for weight class $i$}\}|$. In other words, $\epsilon_i(\sigma, \tau)$ is the number of \emph{surprises} within class $i$, i.e., the number of times some page $\sigma_s$ arrives although some other page $\sigma_t$ of the same class was expected earlier. Let
\[\epsilon(\sigma, \tau) := \sum_i w_i \cdot \epsilon_i(\sigma, \tau).\]
We drop $\sigma, \tau$ from the notation when it is clear from context and bound the competitive ratio of our algorithms in terms of $\epsilon$.  Since  
$\epsilon \le 2\eta$~\cite[Lemma 4.1]{Rohatgi19}, our bounds hold for the $\eta$ error measure as well. In fact, the relationship holds even if $\epsilon_i$ is defined as the \emph{total number of inversions} within weight class $i$ and thus our notion of $\epsilon$ can be significantly smaller than $\eta$ (see \cite{EmekKS21} for an example where $\eta=\Omega(T)\cdot\epsilon$).

\subsection{Our results}
We obtain algorithmic results for learning-augmented weighted paging, both for the case of perfect predictions and for predictions with error. Even though the latter setting generalizes the former, we describe the results separately as most of the key new ideas are already needed for perfect predictions. To the best of our knowledge, no bounds better than $O(k)$ and $O(\log k)$ were previously known even for the case of $\ell=2$ weight classes with perfect predictions.

We first consider the deterministic and the randomized settings when the predictions are perfect.
\begin{theorem}
\label{thm:det}
There is an $\ell$-competitive deterministic algorithm for learning-augmented weighted paging with $\ell$ weight classes and perfect predictions.
\end{theorem}
The competitive ratio is  the best possible by the lower bound of~\cite{Jiang2020online} and is $O(\log W)$ if page weights lie in the range $[1,W]$.
Also, notice that the algorithm is
exactly $\ell$-competitive; in particular, for $\ell=1$ we have an optimal algorithm.  Since $\ell=1$ corresponds to the unweighted case, Theorem~\ref{thm:det} can be viewed as generalizing Belady's \FIF algorithm to the weighted case.

Our algorithm is quite natural, and is based on a water-filling (primal-dual) type approach similar to that for the deterministic $k$-competitive algorithm for weighted paging due to Young~\cite{young}. 
Roughly speaking, the algorithm evicts from each weight class at a rate inversely proportional to its weight, and the evicted page is the one whose next arrival is (predicted) farthest in the future for that weight class. %
While the algorithm is natural, the analysis is based on a novel potential function that is designed to capture the next predicted requests for pages.  
The algorithm and its analysis are described in Appendix~\ref{sec:det-upper}.

Furthermore, for $\ell=1$, this gives a new potential-function proof for the optimality of \FIF. 
This new proof seems simpler and less subtle than the standard exchange argument and might be of independent interest; see Appendix~\ref{sec:belady}.
\begin{theorem}
\label{thm:rand}
There is an $O(\log \ell)$-competitive randomized algorithm for learning-augmented weighted paging with $\ell$ weight classes and perfect predictions. 
\end{theorem}
The competitive ratio is the best possible~\cite{AntoniadisCEPS20,Jiang2020online}, and
is $O(\log \log W)$ for page weights in the range $[1,W]$.
This result is technically and conceptually the most interesting part of the paper and requires several new ideas. 

The algorithm splits the cache space into $\ell$ parts, one for each weight class. Within each class, the cached pages are selected according to a ranking of pages that is induced by running several copies of Belady's \FIF algorithm simultaneously for different cache sizes. The key question is how to maintain this split of the cache space over the $\ell$ classes dynamically over time. To get an $O(\log \ell)$ guarantee, we need to do some kind of a 
{\em multiplicative update} on each weight class, however there is no natural quantity on which to do this update\footnote{The standard weighted paging algorithm does a multiplicative update for each page, but this necessarily loses a factor of $\Omega(\log k)$.}. The main idea is to carefully look at the structure of the predicted requests and the recent requests 
and use this to determine the rate of the multiplicative update for each class.  We give a more detailed overview in Section \ref{sec:overview}.

Both our algorithm and its analysis are rather complicated and we leave the question of designing a simpler randomized algorithm as an interesting open question.

\paragraph{Prediction errors and robustness.}
The algorithms above also work for erroneous predictions, and their performance degrades smoothly as the prediction error increases. In particular, our deterministic algorithm has cost at most $\ell \cdot  \OPT + 2\ell \epsilon$, and our randomized algorithm has expected cost $O(\log \ell \cdot \OPT + \ell \epsilon) $. (Recall $\ell$ is the number of weight classes, and $\epsilon$ is the weighted number of surprises.) Using standard techniques to combine online algorithms \cite{blum2000line,FiatRR94,AntoniadisCEPS20}, together with the worst case $k$- and $O(\log k)$-competitive deterministic and randomized algorithms for weighted paging, this gives the following results.
\begin{theorem}
\label{thm:det-errors}
There is an $O(\min\{\ell+\ell\epsilon/\OPT,k\})$-competitive deterministic algorithm for learning-augmented weighted paging.
\end{theorem}

\begin{theorem}
\label{thm:rand-errors}
There is an $O(\min\{\log \ell+\ell\epsilon/\OPT,\log k\})$-competitive randomized algorithm for learning-augmented weighted paging.
\end{theorem}

\paragraph{Implications for interleaved caching.}
Our algorithms actually only require the relative order of pages {\em within} each weight class, and not how the requests from different classes are interleaved. 
Unweighted paging has also been studied in the {\em interleaved} model, where $\ell$ request sequences $\sigma^{(1)}, \ldots, \sigma^{(\ell)}$ are given in advance and the adversary interleaves them arbitrarily. Here tight $\Theta(\ell)$ deterministic and $\Theta(\log \ell)$ randomized competitive algorithms are known ~\cite{barve2000application,cao1994application, kumar2019interleaved}.
Our results thus extend these results to the weighted setting, where each sequence has pages of a different weight.

\subsection{Overview of techniques}
\label{sec:overview}

We now give a more detailed overview of our 
algorithms. We mainly focus on the case of perfect predictions, and briefly remark how to handle errors, towards the end.  As we aim to obtain guarantees as a function of $\ell$ instead of $k$, the algorithm must consider dynamics at the level of weight classes in addition to that for individual pages.
Our algorithms have two components: a \emph{global strategy} and a \emph{local strategy}. The global strategy decides at each time $t$ how many cache slots to dedicate to each different weight class $i$, denoted $x_i(t)$. Since we have $k$ cache slots in total, we maintain $\sum_i x_i(t) =k$ with $x_i(t) \geq 0$ for all $i$.\footnote{
The $x_i(t)$ may be fractional, but let us assume that they are integral for now.
} The local strategy decides,
for each weight class $i$, which $x_i(t)$ pages to keep in the cache.

Suppose page $\sigma_t$ requested at time $t$ belongs to weight class $r$, and $\sigma_t$ is fetched as it is not in the cache. This increases $x_r(t)$, the number of pages of class $r$ in cache, and the global strategy must decide how to decrease $x_i(t)$ for each class $i \neq r$ to maintain $\sum_i x_i(t)=k$. 
In the deterministic case, roughly, the global strategy simply decreases the $x_i(t)$ uniformly at rate $1/w_i$ (some care is needed to ensure that $x_i(t)$ are integral, and the idea is implemented using a water-filling approach), and the local strategy is Belady's \FIF algorithm.

This suffices for a competitive ratio of $\ell$, but to get an $O(\log \ell)$ bound in the randomized case, one needs more careful {\em multiplicative} updates for $x_i(t)$.
However, it is not immediately clear how to do this and naively updating $x_i(t)$ in proportion to, for example $x_i(t)/w_i$ or $(k-x_i(t))/w_i$ (analogous to algorithms for 
standard weighted paging), does not work.

\paragraph{Update rule.}
A key intuition behind our update rule is the following example. Suppose that for each class $i$, the adversary repeatedly requests pages from some fixed set $P_i$, say in a cyclic order. Assuming $|P_i| \geq  x_i$, 
we claim that the right thing to do is to update each $x_i$ multiplicatively in proportion to
$|P_i|-x_i$ (and inversely proportional to $w_i$). 
Indeed, if $|P_i|$ is already much larger than $x_i(t)$, the algorithm anyway has to pay a lot when the pages in $P_i$ are requested, so it might as well evict more aggressively from class $i$ to serve requests to pages from other classes.  
On the other hand, if $x_i(t)$ is close to $|P_i|$, then the algorithm should reduce $x_i(t)$ at a much slower rate, since it is already nearly correct.

The difficulty in implementing this idea is that the request sequence can be completely arbitrary, and there may be no well-defined {\em working set} $P_i$ of requests for class $i$. Moreover, even if the requests have such structure, the set $P_i$ could vary arbitrarily over time.

Our key conceptual and technical novelty is defining a suitable notion of $P_i$. 
The definition itself is somewhat intricate, but allows us to maintain a ``memory'' of recent requests by utilizing a subset of the real line. (A formal description appears in Section~\ref{sec:mainAlg}.)
Our definition relies on a crucial notion of page ranks that we describe next.

\paragraph{Page ranks.}
Let us fix a weight class $i$, and consider the request sequence $\sigma$ restricted to this class. We say that page $p$ has %
\emph{rank} $m$ at time $t$ if Belady's algorithm running on $\sigma$ with a cache of size $m$ contains $p$ at time $t$, but an alternate version of Belady's algorithm with a cache of size $m-1$ does not. 
The %
rank of pages changes over time, e.g.,~a requested page  always moves to rank $1$ in its weight class.
In Section \ref{sec:branking}, we  describe various properties of this ranking.

Page ranks allow us to define a certain canonical local strategy. %
More importantly, they allow us to view the problem in a clean geometric way,
where the requests for pages correspond to points on the line.
In particular, if the requested page has rank $m$, we think of the request arriving at point $m$ on the line.
Even though the page request sequence can be arbitrary, the resulting {\em rank sequences} in the view above are not arbitrary but have a  useful ``repeat property'', which we crucially exploit in both designing the update rule for $P_i$ and analyzing the algorithm.  (Prediction errors are incorporated quite directly in the above approach, and require only an accounting of how these errors affect the ranks and the repeat property.)

The overall algorithm is described in Section~\ref{sec:mainAlg} and the analysis is described in Section \ref{sec:paginganalysis}.
The analysis uses several potential functions in a careful way. In particular, besides a relative-entropy type potential to handle the multiplicative update of $x_i$, we use additional new potentials to handle the 
dynamics and evolution of $P_i$.

\subsection{Other related work}

Due to its relevance in computer systems and the elegance of the model, several variants of paging have been  studied~\cite{BansalBN12,borodin1998online,Irani96competitiveanalysis}.
An important direction has been to consider finer-grained models and analyses techniques to circumvent the sometimes overly pessimistic nature of worst-case guarantees.
In particular, several {\em semi-online} models, where the algorithm has some partial knowledge of the future input, have been studied, including paging with locality of reference~\cite{borodin1991competitive,irani1996strongly,fiat1997truly}, paging with lookahead~\cite{albers1997influence,breslauer1998competitive,young1991competitive}, Markov paging~\cite{markovpaging}, and interleaved paging \cite{barve2000application,cao1994application,kumar2019interleaved}.
Paging algorithms have also been explored using alternative notions of analysis such as loose-competitivenes~\cite{young}, diffuse adversaries~\cite{koutsoupias2000beyond}, bijective analysis~\cite{angelopoulos2007separation}, and parameterized analysis~\cite{albers2005paging}.

Learning-augmented algorithms have received a lot of attention recently.
Besides paging~\cite{LV18}, they have been considered for a wide range of problems such as ski-rental~\cite{gollapudi2019online,KPS18},  scheduling \cite{BamasMRS20,im2021non,KPS18,mitzenmacher2020scheduling}, load balancing \cite{LattanziLMV,li2021online}, secretary \cite{Secretary}, metrical task systems \cite{AntoniadisCEPS20}, set cover \cite{BamasMS20}, flow and matching \cite{flow-prediction}, and bin packing \cite{angelopoulos2021online}.

In \cite{Jiang2020online}, a different prediction model for weighted paging was considered where at each time, the algorithm has access to a prediction of the \emph{entire} request sequence until the time when every page is requested at least once more. We note that this requires much more predicted information than our model and the analogous models for unweighted paging.

\section{Page ranks, trustful algorithms and repeat violations}
\label{sec:branking}
As discussed in Section \ref{sec:overview}, our algorithm will have two parts: a global strategy and a local strategy. At each time $t$, the global strategy specifies the cache space $x_{i}(t)$ for each weight class $i$ and the local strategy decides which pages of class $i$ to keep. In this section, we define a notion of \emph{ranks} for pages within each class. This will allow us to not only define a local strategy for each class that is close to optimal given any global strategy, but also view the weighted paging problem with arbitrary request sequences in a very clean way in terms of  what we call {\em rank sequences}.

\subsection{Page ranks}\label{sec:ranks}
Fix a weight class $i$.
We define a notion of time-varying {\em ranks} among pages of class $i$. 
Let $\sigma|_i$ be the actual request sequence and let $\tau|_i$ be the sequence of predictions, restricted to class $i$. 
We can view this as an input for unweighted paging. %
For brevity, we use $\sigma = \sigma|_i, \tau = \tau|_i$. 

Let $\Blind(m)$ be the variant of Belady's algorithm for cache size $m$ that, upon a cache miss, evicts the page with the farthest-in-future \emph{predicted} next arrival time (breaking ties arbitrarily, but consistently for all $m$). Note that if all the predictions in $\tau$ are accurate, this is simply Belady's algorithm.
For class $i$, let $C^\position_{i, t}(\sigma,\tau)$ be the set of pages in the cache of $\Blind(m)$ at time $t$; we call $C^\position_{i, t}(\sigma,\tau)$ a \emph{configuration} or \emph{cache state}. We may drop $i$ and $\sigma$ and/or $\tau$ from the notation, and  
assume that $|C^\position_t|=\position$ for all $t$.
A simple inductive argument, whose proof we defer to Appendix \ref{app:missing-ranking-proofs}, shows that the configurations $C^\position_t$ of $\Blind(m)$ differ in exactly one page for consecutive values of $m$:

\begin{restatable}[Consistency]{lemma}{lemmaconsistency} 
\label{lem:consistency}
Let $C^1_0, C^2_0,\ldots$ be any initial configurations, satisfying $C^\position_0\subset C^{\position+1}_0$ for all $\position$.
Then for any sequences $\sigma, \tau$, for all times $t$ and all $\position \geq 0$, we have  $C^{\position}_{t}(\sigma,\tau) \subset C^{\position+1}_{t}(\sigma,\tau)$.
\end{restatable}

Intuitively, Lemma~\ref{lem:consistency} simply says that the set of items in the cache  (of size m) when running $\Blind(m)$ will be a subset of the items when running $\Blind(m+1)$ on a cache of size $m+1$. It
leads to the following well-defined notion of rank%
\footnote{This is very different from the ordering of pages according to their predicted next request time.} on pages at any time $t$.

\begin{Definition}[Page rank]
A page $p$ has {\em rank} $m$ at time  $t$ if $C^\position_t \setminus C^{\position-1}_t=\{p\}$.
\end{Definition}

We now describe how the ranks change when a page is requested. Suppose the requested page $p$ had rank $m_0$ just before it was requested. Then $p$ will have new rank $1$ as it lies in the cache of $\Blind(m)$ for every $m\geq 1$. The pages with ranks $m>m_0$ do not change. Consider the set of pages with (old) ranks $1,\ldots,m_0-1$. (Note that this is precisely $C^{m_0-1}_t$.) Among those pages, the one whose next predicted request is farthest in the future will be updated to have a new rank of $m_0$; denote its original rank by $m_1$. All pages with rank between $m_1$ and $m_0$ will keep their ranks. 
Continuing this way, if we consider the pages of ranks $1, 2, \ldots, {m_1}-1$ (corresponding to $C^{m_1-1}_t$), the page among those whose predicted request appears farthest in the future will have a new rank of $m_1$; denote its original  rank  by $m_2$. We can 
recursively define $m_3$, $m_4$ and so on in a similar fashion.
See Figure~\ref{fig:BRUpdate} for an illustration.
 More formally, we have the following lemma.

\begin{lemma}[Rank update]\label{lem:orderingUpdate}
For a given time $t$, let $p_\position$ denote the page with rank $\position$, and let $\position_0$ be the rank of the next requested page $p_{\position_0}$. Starting from $\position_0$, define the sequence $\position_0>\position_1>\dots>\position_b=1$ inductively as follows: given $\position_a$ for $a\geq 0$, $\position_{a+1}$ is the rank of the page
in $C^{\position_a-1}_t = \{p_1,p_2,\dots,p_{\position_a}-1\}$
with predicted next request farthest in the future. 
If $\position_a=1$, then $b:=a$ and
the sequence ends.
 
Then at time $t+1$, page $p_{\position_0}$ will have rank $1$, and for $a=1,\dots,b$, page $p_{\position_a}$ will have rank $\position_{a-1}$. All other ranks remain unchanged.
\end{lemma}

\begin{proof}
Clearly $p_{\position_0}$ will receive rank $1$ as it must lie in the cache of $\Blind(1)$. Moreover, as $\Blind(m)$ for $\position\ge \position_0$ does not incur a cache miss, ranks greater than $\position_0$ do not change.

Next, by definition of $\position_a$, the page evicted by $\Blind(m)$ for $\position\in\{\position_a,\position_a+1,\dots,\position_{a-1}-1\}$ is $p_{\position_a}$, so $p_{\position_a}$ will have new rank $\ge\position_{a-1}$. 
However, as
 $\Blind(m)$ for $\position\ge \position_{a-1}$ keeps $p_{\position_a}$ in its cache (as it evicts another page), $p_{\position_a}$ will have new rank exactly $\position_{a-1}$.
Also as no page other than $p_{\position_0},\dots,p_{\position_b}$ will be loaded or evicted by any $\Blind(m)$, none of these other pages' ranks will change.
\end{proof}

\begin{figure}[h]
    \centering
    \includegraphics{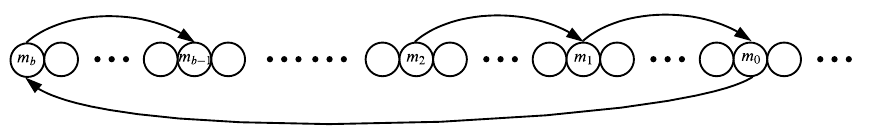}
    \caption{Change of ranks when rank $\position_0$ is requested.  $\position_{a+1}$ is the rank in $\{1,\dots,\position_a-1\}$ whose associated page has the latest predicted next request time, and $\position_b=1$.}
    \label{fig:BRUpdate}
\end{figure}

\subsection{Local strategy and trustful algorithms}\label{sec:local}
Recall that a global strategy specifies a vector $x(t) = (x_{1}(t),\dots,x_{\ell}(t))$ at each time $t$, describing the cache space for each weight class.
For a weight class $i$, consider the local strategy that keeps the pages with ranks between $1$ and $m=x_{i}(t)$ in cache at time $t$.
Note that this is precisely the set $C^\position_{i,t}$.
If $m$ is fractional, 
let us extend the definition of $C^m_{i,t}$ to be the cache state that fully contains the pages with ranks $1,\ldots,\floor m$, and an $m - \lfloor m \rfloor$ fraction of the page with rank $\lfloor m \rfloor + 1$.
We call such algorithms trustful, defined next.

\begin{Definition}[Trustful algorithm]
A weighted paging algorithm is \emph{trustful} if at each time $t$ its configuration (i.e. cache state) is of the form $\bigcup_{i=1}^\ell C_{i,t}^{x_{i}(t)}$ for some $(x_{1}(t),\dots,x_{\ell}(t))$.
\end{Definition}
Next, we show that we can restrict our attention to trustful algorithms\footnote{Note that a trustful algorithm is \emph{not} the same as running \Blind on each weight class (which is optimal for perfect predictions~\cite{peserico2019paging,bender2014cache}). E.g.,~ if $x_{i}(t)$ reduces from $\position$ to $\position-1$, a trustful algorithm will evict the page with rank $\position$, whereas \Blind would evict the page with the farthest predicted next arrival.} without loss of generality.

\begin{restatable}{lemma}{trustfullemma}
\label{lem:optInCi}
Let $A$ be an arbitrary (offline) weighted paging algorithm, and let $x_{i}(t)$ denote the cache space used by class $i$ at time $t$ under $A$. For any arbitrary $\sigma,\tau$, let $A^*$ be the trustful algorithm with configuration $\bigcup_{i=1}^\ell C_{i,t}^{x_{i}(t)}(\sigma,\tau)$ at any time $t$. Then,
	$	\cost_{A^*}(\sigma,\tau)\le 3\cdot\cost_{A}(\sigma)+\epsilon(\sigma, \tau)+O(1)$.
\end{restatable}
The proof of Lemma \ref{lem:optInCi} is in Appendix \ref{app:missing-ranking-proofs} and uses how ranks change over time. As we remark there, a minor modification also yields a much simpler proof of a bound of Wei~\cite{Wei20} for learning-augmented unweighted paging, whose proof was based on an analysis of eleven cases.

Thanks to Lemma~\ref{lem:optInCi}, we can assume that the \emph{offline} algorithm is trustful at the expense of misjudging its cost by an $O(1+\epsilon/\OPT)$ factor.  Note that given a global strategy that computes $x(t)$ online, the corresponding trustful algorithm can be implemented in the learning-augmented setting; the algorithm we design will also be trustful. %

\subsection{Rank sequences and repeat violations}\label{sec:repeat}
Restricting ourselves to trustful algorithms and the local strategy above is useful as we can view the subsequence of page requests for a given weight class as a sequence of \emph{rank requests}.
Consider a single weight class and let $r_1,r_2,\dots$ be the request sequence of pages within that class. 
Then together with sequence $\tau$ of predicted next arrivals, it induces the corresponding \emph{rank sequence}  $h_1,h_2,\dots$ where $h_t$ is the rank of $r_t$ just before it is requested.  A trustful algorithm has a page fault on the arrival of $h_t$ if and only if it has less than $h_t$ pages of that weight class in its cache.

\paragraph{Structure of rank sequences.} It turns out that rank sequences have a remarkable structural property.
In particular,
for the case of perfect predictions ($\epsilon=0$), the possible rank sequences are \emph{exactly} characterized by the following ``repeat property'', which we prove in Appendix~\ref{app:missing-ranking-proofs}.
\begin{restatable}[Repeat property]{lemma}{lemmarepeatproperty}
\label{lem:repeat}
Let $\epsilon=0$ and let $i$ be a weight class. A rank sequence corresponds to a request sequence of pages of class $i$ if and only if it has the following \emph{repeat property}: for any $h$, between any two requests to the same rank $h$, every rank $2,\dots,h-1$ must be requested at least once.
\end{restatable}

With imperfect predictions (i.e., $\epsilon \neq 0$), while there is no clean characterization, the number of times the repeat property is violated can be bounded in terms of $\epsilon$.  We say that time $t$ is a \emph{repeat violation} if rank $h_t$ was also requested at some earlier time (for the same weight class) and some rank from $\{2,\dots,h_t-1\}$ has \emph{not} been requested since then. %
The following lemma bounds the number of repeat violations.
\begin{lemma}[Bounded repeat violations]\label{lem:posSeq}
The number of repeat violations during requests to pages of weight class $i$ is at most $\epsilon_i$.
\end{lemma}

\begin{proof}%
Let $h_1,h_2,\dots$ be the rank sequence of weight class $i$. Consider some time $t_2$ that is a repeat violation of weight class $i$. Let $t_1$ be the last time before $t_2$ that rank $h_{t_2}$ was requested, and let $h:=h_{t_1}=h_{t_2}$. Since $t_2$ is a repeat violation, there exists a rank $h' \in\{2,\dots,h-1\}$ that is missing in the sequence $h_{t_1+1},\dots,h_{t_2-1}$. For any $m$, denote by $p_m$ the page with rank $m$ at a given time. We will show that at time $t_2-1$, page $p_{h'}$ has an earlier predicted arrival time than page $p_h$. Thus, the request to $p_h$ at time $t_2$ increases $\epsilon_i$ by $1$.

Consider the last time in $\{t_1,\dots,t_2-1\}$ when the identity of $p_h$ changes (this time exists as $t_1$ is one such time). By Lemma~\ref{lem:orderingUpdate}, the new page $p_h$ will have the farthest predicted next arrival among the pages $p_2,\dots,p_h$. It suffices to show that it remains true until time $t_2-1$ that $p_h$ has a farther predicted next arrival time than $p_{h'}$. This could only change if $p_{h'}$ changes. But as rank $h'$ itself is not requested between times $t_1$ and $t_2$, by Lemma~\ref{lem:orderingUpdate}, $p_{h'}$ can only change due to a request to a page with larger rank, and the predicted next arrival time of the new page $p_{h'}$ can only decrease due to this change.%
\end{proof}
\section{Algorithm}\label{sec:mainAlg}

We now describe and analyze an $O(\log \ell)$-competitive algorithm for learning-augmented weighted paging without prediction errors (Theorem~\ref{thm:rand}), and more generally show that it is $O(\log \ell+\ell\epsilon/\OPT)$-competitive with imperfect predictions. Combining our algorithm with any $O(\log k)$-competitive weighted paging algorithm via the combination method of Blum and Burch \cite{blum2000line} yields Theorem~\ref{thm:rand-errors}.

By Lemma~\ref{lem:optInCi} we can assume that both the online and offline algorithms are trustful, and hence are fully specified by vectors $(x_1(t),\dots,x_\ell(t))$ and $(y_1(t),\dots,y_\ell(t))$ that describe the cache space used by each weight class under the online and offline algorithm, at each time $t$.
For the online algorithm, we allow $x_i(t)$ to be fractional so that the cache contains pages with ranks $1,\dots,\floor{x_i}$ fully and a $x_i - \floor{x_i}$ fraction of the page with rank $\floor{x_i}+1$.
By standard techniques~\cite{bansal2012primal}, this can be converted into a randomized algorithm with integer $x_i(t)$, losing only a constant factor in the competitive ratio.
 
By the rank sequence view in Section \ref{sec:repeat}, the problem can be restated as follows. At each time $t$, some rank $p$ in some weight class $r$ is requested. If $x_r(t)<p$, there is a cache miss and the algorithm incurs cost $w_r\cdot\min(1, p-x_r(t))$. The goal is to design the update of $x_i(t)$ so as to minimize the total cost. Before giving the details, we first present an overview of the algorithm.

\subsection{Overview}

For convenience, we usually drop the dependence on time $t$ from the notation.  Consider a request to a page with rank $p$ in weight class $r$ and suppose our algorithm has a cache miss, i.e.,~$x_r<p$.  Then $x_r$ always increases and all other $x_i$ with $i\neq r$ and $x_i>0$  decrease. 
So the main question is at what rate to decrease the $x_i$ for $i\neq r$ (since $\sum_i x_i=k$ at all times, the rate of increase  $x_r'=-\sum_{i\neq r} x_i'$).
A key idea is to have a variable $\mu_i$ for each weight class $i$ that is roughly proportional to the eviction rate from weight class $i$. This variable also changes over time depending on the ranks of requested pages within each weight class. We now describe the main intuition behind how the $\mu_i$ are updated.

For each weight class $i$, we maintain a more refined memory of the past in the form of a set $S_i\subseteq [x_i,\infty)$, consisting of a union of intervals. %
Roughly speaking, the set $S_i$ can be thought of as an approximate memory of ranks of weight class $i$ that were requested relatively recently. So the $P_i$ in Section \ref{sec:overview} corresponds to $(0,x_i] \cup S_i$.
We set $\mu_i = |S_i|$, the Lebesgue measure of $S_i$.

For weight classes $i\ne r$, the algorithm increases $\mu_i$ multiplicatively, so that eventually we will evict faster from weight classes that are only rarely requested. For the same reason, for the requested weight class $r$ we would like to decrease $\mu_r$ (as we increased $x_r$, we would not want to decrease it rapidly right away when requests arrive next in other classes). However, decreasing $\mu_r$ could be highly problematic in the case that the offline algorithm has $y_r<x_r$, because this will slow down our algorithm in decreasing $x_r$ in the future, making it very difficult to catch up with $y_r$ later.

To handle this, crucially relying on  the repeat property, we decrease $\mu_r$ if and only if $p\in S_r$.
Informally, the reason is the following: if the last request of rank $p$ was recent, then---assuming no repeat violation---all the ranks in $\{2,\dots,p\}$ were also requested recently. This suggests it may be valuable to hold $p$ pages from weight class $r$ in cache, but our algorithm only holds $x_r<p$ of them. To improve our chance of increasing $x_r$ towards $p$ in the future, we should reduce $\mu_r$.

This simplified overview omits several technical details and in particular how to update the sets $S_i$. We now describe the algorithm in detail and then give the analysis in Section~\ref{sec:paginganalysis}.

\subsection{Detailed description}
As stated above, we denote by $x_i\in[0,k]$ the page mass from weight class $i$ in the algorithm's cache, meaning that the pages with ranks $1,\dots,\floor{x_i}$ of weight class $i$ are fully in its cache and the next page is present to an $x_i - \floor{x_i}$ extent. We maintain a variable $\mu_i$ and a set $S_i\subset[x_i,\infty)$ for weight class $i$. The $x_i, \mu_i, S_i$ are all functions of time $t$, but we suppress this dependence for notational convenience. Algorithm~\ref{alg:randInterleaved} contains a summary of our algorithm, that we now explain in detail.

\begin{algorithm2e}[h!]
\caption{Serving a request to rank $q$ of weight class $r$.}
\label{alg:randInterleaved}
\DontPrintSemicolon
Initialize $p:=q-1$\;
\While{$p<q$}{
    Increase $p$ at rate $p'=8$\;
    \If{$p>x_r$}{
        Update $x$ and $\mu$ according to \eqref{eq:pagingUpdateX}, \eqref{eq:pagingUpdateRho}\;
        \For{$i\ne r$}{
            Add to $S_i$ the points that $x_i$ moves past\;
        }
        Remove points from the left and right of $S_r$ at rates $x_r'$ and $1$\;
        \If{$p\notin S_r$}{
            Add points from $(q-1,p]\setminus S_r$ to $S_r$ at rate $2$\;
        }
    }
}
\end{algorithm2e}

\paragraph{Request arrival and continuous view.}
Consider a request to the page with rank $q$ of weight class $r$. If $x_r\le q-1$ (resp.~$x_r \geq q$), the page is fully missing (resp.~fully present) in the algorithm's cache. But for $q-1<x_r<q$, the page is fractionally present. To obtain a view where each request is either fully present or fully missing, we break the request to rank $q$ into infinitesimally small requests to each point $p\in(q-1,q]$, by moving a variable $p$ continuously along the interval $(q-1,q]$. We call $p$ the \emph{pointer} (to the current request) and move it from $q-1$ to $q$ at speed $8$, so the duration of this process is $1/8$ unit of time. During this process, we update $x$, $\mu$, and the sets $S_i$ continuously.

\paragraph{Update of $x$ and $\mu$.}
Let $\delta:=\frac{1}{\ell}$ and define variables \[M:=\sum_{i=1}^\ell\mu_i, \qquad  \beta_i = \frac{\mu_i+\delta M}{w_i M}, \qquad \text{ and } \qquad  B:=\sum_{i=1}^\ell \beta_i.\]
As $\mu$ changes over time, these quantities also change over time.

While $p\le x_r$, the algorithm does nothing (the corresponding infinitesimal request is already present in its cache). While $p>x_r$, we update $x_i$ and $\mu_i$, for each class $i$, at rates
\begin{align}
x_i'&= \1_{\{i=r\}} - \frac{\beta_i}{B}\label{eq:pagingUpdateX}\\
\mu_i'&=\frac{\beta_i}{B} - 2\cdot \1_{\{i=r\text{ and }p\in S_r\}}.\label{eq:pagingUpdateRho}
\end{align}
Let us make a few observations. Each $\beta_i$ lies between $0$ and $B$, and hence $x_r$ increases while all $x_i$ for $i\neq r$ decrease.
Moreover, we have that $\sum_{i} x_i'=0$, and so the total cache space $\sum_i x_i$ used by our algorithm remains constant at its initial value of $k$. Second, $\mu_i$ increases for all $i \neq r$. For $i=r$, we decrease $\mu_r$ if and only if $p\in S_r$.

\paragraph{Update of the sets $S_i$.} Each set $S_i$ maintained by our algorithm will be a (finite) union of intervals satisfying

(i) $S_i\subset [x_i,\infty)$ and

(ii) $|S_i|=\mu_i$,

\noindent where $|\cdot|$ denotes the Lebesgue measure. We initialize $S_i$ arbitrarily to satisfy these properties (e.g., as $S_i:=[x_i,x_i+\mu_i)$). %

The precise update rule for the sets $S_i$ is as follows. 
First consider $i\ne r$.
Here $x_i$ decreases, and we simply add to $S_i$ all the points that $x_i$ moves over. 
As $S_i\subset[x_i,\infty)$,
these points are not in $S_i$ yet, so $S_i$ grows at rate $-x_i' = \beta_i/B$. As $\mu_i'$ is also $\beta_i/B$, both $|S_i|=\mu_i$ and $S_i\subset [x_i,\infty)$ remain satisfied. 

For the requested weight class $r$, modifying $S_r$ is somewhat more involved, consisting of three simultaneous parts (see Figure~\ref{fig:SrUpdate}):
\begin{itemize}[nosep]
    \item We remove points from the left (i.e., minimal points) of $S_r$ at the rate at which $x_r$ is increasing (i.e., we increase the left boundary of the leftmost interval of $S_r$ at rate $x_r'$). This ensures that the property $S_r\subset[x_r,\infty)$ is maintained.
    \item We also remove points from the right of $S_r$ at rate $1$. We can think of these points as ``expiring'' from our memory of points that were recently requested.
    \item While $p\in S_r$, we do nothing else. Observe that $|S_r|=\mu_r$ remains satisfied.
    \item While $p\notin S_r$, we also add points from $(q-1,p]\setminus S_r$ to $S_r$ at rate $2$, thereby ensuring again that $|S_r|=\mu_r$ remains satisfied. This can be achieved as follows, which will also ensure that $S_r$ continues to be a finite union of intervals: Consider the leftmost interval of $S_r$ that overlaps with $(q-1,p]$; if no such interval exists, consider the empty interval $(q-1,q-1]$ instead (or, if $x_r\in(q-1,q]$, consider the empty interval $(x_r,x_r]$). Since $p\notin S_r$, the right boundary of this interval lies in $[q-1,p]$. We can add points from $(q-1,p]\setminus S_r$ to $S_r$ at rate $2$ by shifting this boundary to the right at rate $2$. (Since $p$ moves at rate $8$, the boundary will never ``overtake'' $p$.) If the considered interval happens to be the rightmost interval of $S_r$, then combined with the previous rule of removing points from the right at rate $1$ this would mean that effectively the right boundary moves to the right only at rate $1$ instead of $2$.
\end{itemize}
Note that the removal of points from the left and right of $S_r$ is always possible: If $S_r$ were empty, then $p\notin S_r$ and the addition of points at rate $2$ is faster than the removal.

\begin{figure}
    \centering
    \includegraphics{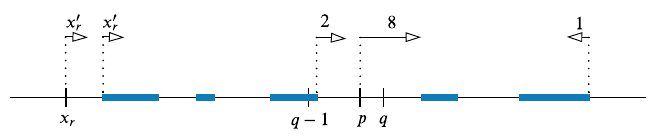}
    \caption{Update of $S_r$ (blue) when $p\notin S_r$}
    \label{fig:SrUpdate}
\end{figure}

\paragraph{Boundary conditions.} The description of the algorithm is almost complete, except that it does not ensure yet that the requested page is fully loaded to its cache, which requires that $x_r\ge 1$ (recall that the requested page will receive new rank $1$), and that no $x_i$ decreases below $0$. Both issues can be handled as follows. Add a dummy page $d_i$ to each weight class $i$ and increase the cache size by $\ell$, where this extra cache space shall be used to hold the $\ell$ dummy pages. In the request sequence, replace any request to a page $v$ by many repetitions of the request sequence $(v,d_1,\dots,d_\ell)$. As the number of repetitions tends to infinity, our algorithm converges to a state with $x_i\ge 1$ for all $i$ and $x_r\ge 2$ for the weight class $r$ containing $v$. Thus, it holds $v$ and all dummy pages in its cache. The corresponding algorithm obtained by removing the dummy pages (and reducing the cache size back to $k$) is a valid paging algorithm for the original request sequence.

\section{Analysis}
\label{sec:paginganalysis}

The goal of this section is to prove the following theorem:
\begin{theorem}\label{thm:randNonrobust}
The algorithm is $O(\log\ell+\ell\epsilon/\OPT)$-competitive for learning-augmented weighted paging.
\end{theorem}

We employ a potential function based analysis.
Our potential function $\overpot$ will consist of three potentials $\entpot, \reppot$, and $\scatpot$ and is defined as follows.  
\begin{align*}
 &\text{(Entropic potential)} &&\entpot := \sum_{i} w_i \left[5[\mu_i+(x_i-y_i)_+]\log\frac{(1+\delta)[\mu_i+(x_i-y_i)_+]}{\mu_i+\delta [\mu_i+(x_i-y_i)_+]} + \left[\mu_i + 2(x_i - y_i)\right]_+\right], \\
     &\text{(Repeat potential)}  && \reppot:= \sum_i w_i\int_{S_i}\frac{|(y_i,x_i]\cap R_{i u}|}{\mu_i+\delta(\mu_i+|(y_i,x_i]\cap R_{i u}|)} du,\\
  & \text{(Scatter Potential)} && \scatpot:= \sum_{i} w_i\left(|S_i\cap[y_i,\infty)|+(x_i-y_i)_+\right),\\
&\text{(Overall potential)} &&\overpot := 2 \entpot + 5 \reppot + 4 \mathcal {S}.
\end{align*}
Here $y_i$ describes the state of a \emph{trustful} offline algorithm (whose cost is within a factor $O(1+\epsilon/\OPT)$ of the optimal offline algorithm by Lemma~\ref{lem:optInCi}). We use $(x_i - y_i)_+ := \max \{0, x_i - y_i\}$. The set $R_{iu}$, for a weight class $i$ and $u\geq 0$, consists of all points that were visited by the pointer since the last request at $u$ (details in Section \ref{sec:offCostPot}). In the definition of the repeat and scatter potentials, $|\cdot|$ denotes the Lebesgue measure.

The potential $\entpot$ bears similarities with entropy/Bregman divergence type potentials used in other contexts~\cite{bansal2010simple,BubeckCLL19,BubeckCLLM18,BuchbinderGMN19,CoesterL19}. %
The potential $\reppot$ is carefully designed to exploit the repeat property. In particular, just before the pointer reaches $u$, the set $R_{iu}$ contains $[1,u)$ provided there is no repeat violation, and it becomes empty immediately afterwards. The scatter potential is mostly for technical reasons to handle that the intervals in $S_i$ may be non-contiguous. 

To show that our algorithm has the desired competitive ratio, it suffices to show that
\begin{align}
    \on' + \overpot' \leq O(\log \ell) \cdot \Off'+O(\ell)\epsilon'\label{eq:potMainBound}
\end{align}
at all times, where $\on$, $\Off$, and $\epsilon$ are appropriate continuous-time notions of online cost, offline cost, and the prediction error and $a'$ denotes the derivative of $a$ with respect to time. We will actually take $\text{On}$ to be a quantity we call \emph{online pseudo-cost} that approximates the true online cost, as discussed in Section~\ref{sec:pseudo}. %
For each request, we will consider the step where the offline algorithm changes its vector $(y_1,\dots,y_\ell)$ separately from the remaining events. %

\subsection{Offline cost}
\label{sec:offlinecost}
We will charge the offline algorithm for both the weights of pages it fetches as well as pages it evicts.\footnote{Note that over all time steps, the total fetching and total eviction cost are within an additive constant of each other.} We may assume that for each request, the (trustful) offline algorithm proceeds in two steps: First it changes $y$ and updates its cache content with respect to the \emph{old} ranks. Then it updates its cache content to reflect the new ranks of weight class $r$. This means that it may fetch a page of weight class $r$ in the first step that it evicts again in the second step in order to load the requested page (recall Lemma~\ref{lem:orderingUpdate}/Figure~\ref{fig:BRUpdate}). Since both of these pages have the same weight, this overestimates the offline cost by only a constant factor.

\paragraph{Change of $\boldsymbol{y}$.} When the offline algorithm changes some $y_i$ at rate $y_i'$, it incurs cost $w_i|y_i'|$.

We claim that the term corresponding to weight class $i$ in each of the potentials changes at rate at most $O(\log \ell)\,w_i |y_i'|$.
First, for the entropic potential, it is not hard to see that $\frac{d\entpot}{d y_i}=O\!\left(1+\log\frac{1+\delta}{\delta}\right)w_i=O(\log \ell)w_i$. Therefore,
\[\entpot' \leq O(\log \ell)\, w_i |y_i'|.\] Next, for the scatter potential, 
as $[y_i,\infty)$ and $(x_i-y_i)_+$ can change at rate at most $|y_i'|$, we also have that $\scatpot' \leq 2 w_i |y_i'|$. It remains to bound $\reppot'$.

Fix a $u \in S_i$ and consider the term in the integrand of $\reppot$ corresponding to $u$.
If $y_i$ increases, then $(y_i,x_i]\cap R_{iu}$ can only decrease, in which case $\reppot$ only decreases.
If $y_i$ decreases, then $(y_i,x_i]\cap R_{iu}$ increases at rate at most $|y_i'|$. Ignoring the increase in the denominator (which only decreases $\reppot$), the increase in the numerator leads to an increase of at most
\[    \frac{ |y_i'|}{\mu_i + \delta(\mu_i +|(y_i,x_i]\cap R_{iu}|)} \leq  \frac{ |y_i'|}{\mu_i}. \]
As $|S_i|=\mu_i$, integrating the above over points in $S_i$, the integral rises at rate at most $|S_i| |y_i'|/\mu_i = |y_i'|$.
So, $\reppot' \leq w_i |y_i'|$ as desired.

To summarize, when $y$ changes, \eqref{eq:potMainBound} holds as we can charge the change in the overall potential to $O(\log \ell)$ times the offline cost.

\paragraph{Change of ranks.} We can assume that the offline vector $y$ is integral. %
If $y_r <q$, then offline pays $w_r$ to fetch the page. As the pointer $p$ moves in $[q-1,q]$ at rate $8$,
equivalently, we can view this as charging the offline algorithm continuously at rate $\1_{\{y_r<p\}}8w_r$ during the movement of $p$. This view will be useful for analyzing the online algorithm in a continuous way, as we do next.

\subsection{Online pseudo-cost}
\label{sec:pseudo}
Instead of working with the actual cost incurred by the online algorithm, it will be convenient to work with a simpler \emph{online pseudo-cost}. This is defined as the quantity that is initially $0$ and grows at rate $1/B$ at all times during which the online algorithm changes $x$. 
 Recall that $B$ is itself changing over time. Formally, we define the online pseudo-cost as 
 \[\int_{0}^\infty \frac{1}{B(t)} \cdot \1_{\{x(t)' \neq 0\}}dt.\]
 The following lemma shows that the online pseudo-cost is a good proxy for the true online cost. %
\begin{lemma}
\label{lem:pseudo}
Up to a bounded additive error, the total cost of the online algorithm is at most $18$ times the online pseudo-cost. 
\end{lemma}
\begin{proof}
The online cost is the weighted page mass loaded to its cache. The two events that lead to page mass being loaded to the cache are either an increase of $x_r$ or a change of ranks.

By definition of ranks (or alternatively, Lemma~\ref{lem:orderingUpdate}/Figure~\ref{fig:BRUpdate}), a request to rank $q$ does not affect the (unordered) \emph{set} of pages with ranks $1,2,\dots,s$ for any $s\ge q$, and for $s<q$ it affects this set only by adding and removing one page. Thus, if rank $q$ of class $r$ is requested, the change of ranks incurs cost $w_r\min\{1,q-x_r\}_+$.

We can overestimate the cost due to increasing $x_r$ by viewing the change of $x_r$ as an increase at rate $1$ separate from a decrease at rate $\beta_r/B$. In this view, the online cost for increasing $x_r$ is $w_r$ times the duration of the increase.
Since the pointer $p$ moves at speed $8$ across $(q-1,q]$, and changes to $x$ occur only while $p>x_r$, the duration of the update of $x$ for this request is precisely $\frac{1}{8}\min\{1,q-x_r\}_+$ (where $x_r$ denotes the value of this variable before the request arrives). Therefore, the online cost for increasing $x_r$ is $\frac{1}{8}w_r\min\{1,q-x_r\}_+$, and hence the total online cost is $9$ times the cost for increasing $x_r$ (in our overestimating view).

Over the course of the algorithm, the overall increase of any $x_i$ equals the decrease, up to an additive constant. Thus, instead of charging for increasing $x_r$ and the change of ranks, we can charge only for decreasing each $x_i$ (including $i=r$) at rate $\beta_i/B$ (the associated cost being $w_i$ times this quantity). This underestimates the true online cost by a factor of at most $9$, up to an additive constant. The cost charged in this way increases at rate
\[ \sum_i w_i  \frac{\beta_i}{B} = \sum_i w_i \frac{\mu_i+\delta M}{B w_i M}=\frac{M+M}{B M}=\frac{2}{B},\]
which is twice the rate of increase of the pseudo-cost, and the result follows.
\end{proof}
By Lemma \ref{lem:pseudo} it now suffices (up to $O(1)$ factors) to assume that $\text{On}' = \frac{1}{B} \cdot \1_{\{x' \neq 0\}}$. Consequently, recalling that the offline algorithm suffers cost at rate $\1_{\{y_r<p\}}8w_r$ due to the change of ranks, in order to prove the desired competitive ratio it suffices to show that \[\frac{1}{B}\cdot \1_{\{x' \neq 0\}} + \overpot' \leq O(\log \ell)\cdot \1_{\{y_r < p\}}w_r +O(\ell)\epsilon'.\]
We now upper bound the rate of change of each of the potentials.

\subsection{\texorpdfstring{Rate of change of the entropic potential $\entpot$}{Rate of change of the entropic potential}}

\begin{lemma}
\label{del-phi-paging}
The rate of change of the entropic potential $\entpot$ is at most
\begin{align*}
    \entpot'\quad\le\quad& -\frac{1}{2B}
    \quad+\quad \1_{\{y_r<p\}}w_r O(\log \ell)\quad-\quad\frac{5}{2}\sum_{i\ne r}\frac{w_i(x_i-y_i)_+\mu_i'}{\mu_i+\delta(\mu_i+x_i-y_i)} \\
    &\qquad+ \1_{\{p\in S_r\}}\frac{15w_r(x_r-y_r)_+}{\mu_r + \delta(\mu_r+x_r-y_r)} \quad+\quad \1_{\{y_r < x_r + \mu_r \text{ and } p \notin S_r\}}2w_r,
\end{align*}
\end{lemma}
Note that the first term $-\frac{1}{2B}$ can be charged against the online pseudo-cost and the second term can be charged to the offline cost resulting from the change of ranks. The third term is negative, which only helps. However, the last two terms are problematic. We will handle them using the repeat potential and the scatter potential. 

\begin{proof}[of Lemma \ref{del-phi-paging}]
Define  $\tM:=\sum_i(\mu_i+x_i-y_i)_+$ and $L:=\sum_i(x_i-y_i)_+=\sum_i(y_i-x_i)_+$, where the equality is due to $\sum_{i=1}^\ell x_i=\sum_{i=1}^\ell y_i=k$.
We note that
\begin{equation}
\label{lem:CL}
    \max\{M,L\} \le \tM \le M+L,
\end{equation}
where the first inequality follows as $M=\sum_i\mu_i=\sum_i (\mu_i + x_i -y_i)\le \tM$ and $L\le \tM$ as $\mu_i \geq 0$ for all $i$, and the second inequality follows from the triangle inequality.

Letting
\[
\Psi:=\sum_{i\colon x_i\ge y_i} w_i (\mu_i+x_i-y_i)\log\frac{(1+\delta)(\mu_i+x_i-y_i)}{\mu_i+\delta (\mu_i+x_i-y_i)}\]
we can write
\[
\entpot=5\Psi+\sum_{i} w_i \left[\mu_i + 2(x_i - y_i)\right]_+.\]
We first bound the change of $\Psi$. We first consider the case $x_r< y_r$, so that the summand for $i=r$ is $0$ and does not contribute to $\Psi'$. For $i\ne r$, the sum $\mu_i+x_i$ is unchanged by \eqref{eq:pagingUpdateX} and \eqref{eq:pagingUpdateRho}, and therefore
\begin{align}
\Psi' &= -\sum_{i\colon x_i\ge y_i}w_i\frac{\mu_i+x_i-y_i}{\mu_i+\delta(\mu_i+x_i-y_i)}\mu_i'\label{eq:phiChangeMainPagingBegin}
\,\,\,= -\sum_{i\colon x_i\ge y_i}\frac{\mu_i+x_i-y_i}{\mu_i+\delta(\mu_i+x_i-y_i)}\frac{\mu_i+\delta M}{B M}  \\
&\le -\sum_{i\colon x_i\ge y_i}\frac{\mu_i+x_i-y_i}{\mu_i+\delta(\mu_i+x_i-y_i)}\frac{\mu_i+\delta \tM}{B \tM}
\,\,\,
\le -\sum_{i\colon x_i\ge y_i}\frac{x_i-y_i}{\mu_i+\delta\tM}\frac{\mu_i+\delta \tM}{B\tM}\nonumber\\
&= -\frac{L}{B\tM} \label{eq:phiChangeMainPaging}
\,\,\, \le -\frac{1}{B}\frac{L}{M+L}\qquad\qquad\text{(if $x_r< y_r$)},
\end{align}
where the first inequality uses $M\le\tM$, and the second inequality uses $\mu_i\ge 0$ and $\mu_i+x_i-y_i \leq \tM$.

We will actually need the following slightly more complicated bound, which can be obtained by combining  \eqref{eq:phiChangeMainPagingBegin} and  \eqref{eq:phiChangeMainPaging}:
\[
    \Psi'\le -\frac{1}{2}\sum_{i\ne r}\frac{w_i(x_i-y_i)_+\mu_i'}{\mu_i+\delta(\mu_i+x_i-y_i)} -\frac{1}{2B}\frac{L}{M+L}\qquad\qquad\text{(if $x_r< y_r$)}.
\]
Thus, the contributions considered so far only lead to a decrease of $\Psi$. However if $x_r\ge y_r$, then $\Psi$ could also suffer an increase resulting from increasing $x_r$ at rate $1$ and, if $p\in S_r$, decreasing $\mu_r$ at rate $2$. In this case, the change of $\Psi$ can exceed the preceding bound by at most
\begin{align}
    &w_r\left[\log\frac{(1+\delta)(\mu_r+x_r-y_r)}{\mu_r+\delta (\mu_r+x_r-y_r)} + 1 + \1_{\{p\in S_r\}}\frac{\mu_r+x_r-y_r}{\mu_r + \delta(\mu_r+x_r-y_r)}(1+\delta)2 \right]
    \nonumber\\   
    & \le w_r \cdot O(\log \ell) + \1_{\{p\in S_r\}}\frac{3w_r(x_r-y_r)_+}{\mu_r + \delta(\mu_r+x_r-y_r)}.\label{eq:phiChangeExtraPaging}
\end{align}
In summary, while $x$ is changing, $\Psi$ is changing at rate
\begin{align}
    \Psi'\le& -\frac{1}{2}\sum_{i\ne r}\frac{w_i(x_i-y_i)_+\mu_i'}{\mu_i+\delta(\mu_i+x_i-y_i)} -\frac{1}{2B}\frac{L}{M+L}
    + \1_{\{y_r<p\}}w_r O(\log \ell) + \1_{\{p\in S_r\}}\frac{3w_r(x_r-y_r)_+}{\mu_r + \delta(\mu_r+x_r-y_r)},\label{eq:PhiFirstPart}
\end{align}
where the term $\1_{\{y_r<p\}}$ comes from the fact that the extra increase \eqref{eq:phiChangeExtraPaging} is incurred only if $y_r\le x_r$ and $x$ is changing only if $x_r<p$.

We now bound the change of the part of $\entpot$ not involving $\Psi$, i.e., of the quantity $\entpot-5\Psi= \sum_{i} w_i\left[\mu_i + 2(x_i - y_i)\right]_+.$
Using the update rules for $\mu$ and $x$, and cancelling some common terms, the rate of change can be written as
\begin{align}
    (\entpot-5\Psi)'&= \1_{\{\mu_r > 2(y_r-x_r)\text{ and }  p\notin S_r\}}2w_r \quad-\quad\;\smashoperator{\sum_{i\colon \mu_i > 2(y_i-x_i)}}\; w_i\frac{\beta_i}{B}\nonumber\\
    &\le \1_{\{y_r < x_r + \mu_r\text{ and } p\notin S_r\}}2w_r \quad-\quad\frac{1}{B M}\;\;\smashoperator{\sum_{i\colon \mu_i > 2(y_i-x_i)}}\; \mu_i.\label{eq:psiChangeBeginPaging}
\end{align}
As $\sum_i \mu_i = M$, the sum in \eqref{eq:psiChangeBeginPaging} can be rewritten as\[
    \;\smashoperator{\sum_{i\colon \mu_i > 2(y_i-x_i)}}\; \mu_i= \Bigg[M\;-\;\;\smashoperator{\sum_{i\colon \mu_i \le 2(y_i-x_i)}}\; \mu_i\;\,\Bigg]_+
    \ge \Bigg[M - \sum_{i: 0 \leq 2(y_i-x_i)} 2(y_i-x_i)\Bigg]_+ 
    = [M  - 2L]_+.
\]
Thus,
\begin{align}
    (\entpot-5\Psi)' \le \1_{\{y_r < x_r + \mu_r \text{ and } p \notin S_r \}}2w_r \quad-\quad\frac{1}{B}\left[1-\frac{2L}{M}\right]_+.\label{eq:PhiSecondPart}
\end{align}
The lemma follows by combining \eqref{eq:PhiFirstPart} and \eqref{eq:PhiSecondPart} and noting that
\begin{align*}
    -\frac{5}{2B}\frac{L}{M+L}\quad - \quad\frac{1}{B}\left[1-\frac{2L}{M}\right]_+ \le -\max\left\{\frac{5}{2B}\frac{L}{M+L}\quad, \quad\frac{1}{B}\left[1-\frac{2L}{M}\right]_+\right\} \le -\frac{1}{2B},
\end{align*}
which can be seen by considering separately the cases $M\le 4L$ and $M> 4L$.
\end{proof}

\subsection{\texorpdfstring{The  repeat potential $\reppot$  and its rate of change}{The  repeat potential and its rate of change}}\label{sec:offCostPot}
The purpose of the repeat potential is to cancel the term
\begin{align*}
\1_{\{p\in S_r\}}\frac{15w_r(x_r-y_r)_+}{\mu_r + \delta(\mu_r+x_r-y_r)},
\end{align*}
from our bound on $\entpot'$ in case the current request is not a repeat violation. 
We will crucially use that if the current request is not a repeat violation, then since the last request to rank $q$ of weight class $r$, every rank less than $q$ has been requested at least once.

For a weight class $i$ and $u\in\R_+$, denote by $R_{i u}$ the set of points across which the pointer $p$ has moved during requests of weight class $i$ \emph{after} the time when the pointer was last located at $u$ for a request to weight class $i$. In other words, after any request $R_{i u}$ is the set $(u,\ceil u] \cup \bigcup_s[s-1,s]$, where $s$ ranges over all ranks of weight class $i$ that have been requested (so far) after the last request to rank $\ceil u$ of weight class $i$. If the pointer was never at $u$ during a request to weight class $i$, define $R_{i u}:=\R_+$ as the entire positive real line.

Recall that the \emph{repeat potential} is defined as
\[   \reppot:=\sum_{i=1}^\ell w_i\reppot_i, \qquad 
    \text{where }\qquad\reppot_i:= \int_{S_i}\frac{|(y_i,x_i]\cap R_{i u}|}{\mu_i+\delta(\mu_i+|(y_i,x_i]\cap R_{i u}|)} du,\]
and $|\cdot|$ denotes the Lebesgue measure.

When online moves, $\reppot_i$ will change due to the changes in the values of $x_i$, $\mu_i$, the sets $S_i$ and $R_{iu}$. We will consider the effect of each of these changes separately while analyzing $\reppot'$.     
We also consider $i\ne r$ and $i=r$ separately.
The case $i\neq r$ is quite simple, and we describe it next.

\paragraph{Change of $\bm{\reppot_i}$ for $\bm{i\ne r}$:}
Since the fraction in the definition of $\reppot_i$ is non-decreasing in $x_i$, and $x_i$ decreases for $i\ne r$, the change of $x_i$ does not cause any increase of $\reppot_i$. Similarly, $\reppot_i$ is non-increasing in $\mu_i$ and $\mu_i$ increases for $i \ne r$, so also the change of $\mu_i$ does not cause any increase of $\reppot_i$. Moreover, as the request is to weight class $r\ne i$, the sets $R_{iu}$ do not change.
However, $\reppot_i$ could increase due to points being added to $S_i$ (recall that as $x_i$ decreases we add the points that $x_i$ moves over to $S_i$). 

The integrand corresponding to any $u\in S_i$ can be bounded by
\begin{align*}
\frac{|(y_i,x_i]\cap R_{i u}|}{\mu_i+\delta(\mu_i+|(y_i,x_i]\cap R_{i u}|)} \le \frac{(x_i-y_i)_+}{\mu_i+\delta(\mu_i+x_i-y_i)}.
\end{align*}
As new points $u$ are added to $S_i$ at rate $\mu_i'$, the change of $\reppot_i$ is bounded by
\begin{align*}
    \reppot_i'\le \frac{(x_i-y_i)_+\mu_i'}{\mu_i+\delta(\mu_i+x_i-y_i)}\qquad\qquad\qquad\text{(for $i\ne r$)}.
\end{align*}

\paragraph{Change of $\bm{\reppot_r}$:} For $i=r$, all the relevant quantities change, and we consider them separately.

\begin{description}
\item[\textnormal{\textit{Effect of changing $x_r$:}}] 
Fix a point $u\in S_r$. The increase of $x_r$ can increase the integrand corresponding to $u$ at most at rate
\[     \frac{x_r'}{\left(\mu_r+\delta(\mu_r+|(y_r,x_r]\cap R_{r u}|)\right)} \le \frac{1}{\mu_r}, \]
as the numerator rises at rate at most $x_r'\le 1$, and the denominator can only rise and reduce the integrand (which we ignore). Note that this increase occurs only if $y_r<x_r<p$, (as $x_r$ does not increase if $p\leq x_r$, and the interval $(y_r,x_r]$ is empty if $y_r \geq x_r$). Thus, as $|S_r|=\mu_r$, the contribution of the change of $x_r$ to $\reppot_r'$ is at most
\[  \1_{\{y_r<p\}}\int_{S_r} \frac{1}{\mu_r} du =\1_{\{y_r<p\}}. \]

\item[\textnormal{\textit{Effect of changing $\mu_r$:}}] 
For a fixed $u$,
the increase in the integrand due to the change in $\mu_r$ is
\begin{align*}
    \frac{-(1+\delta)|(y_r,x_r]\cap R_{r u}| \cdot \mu_r'}{\left(\mu_r+\delta(\mu_r+|(y_r,x_r]\cap R_{r u}|)\right)^2}
    &\le \1_{\{p\in S_r\}}\frac{(1+\delta)2(x_r-y_r)_+}{(1+\delta)\mu_r\left(\mu_r+\delta(\mu_r+x_r-y_r)\right)}\\
    &= \1_{\{p\in S_r\}}\cdot\frac{1}{\mu_r}\cdot\frac{2(x_r-y_r)_+}{\mu_r+\delta(\mu_r+x_r-y_r)}
\end{align*}
where the first inequality uses that $\mu_r'\leq 0$ if and only if $p \in S_r$ and in that case $-\mu_r'\leq 2$. 

As this expression does not depend on $u$ and as $|S_r|=\mu_r$, the change of $\mu_r$ contributes to $\reppot_r'$
\[ \int_{S_r}  \1_{\{p\in S_r\}}\frac{1}{\mu_r}\frac{2(x_r-y_r)_+}{\mu_r+\delta(\mu_r+x_r-y_r)} du \leq  \1_{\{p\in S_r\}}\frac{2(x_r-y_r)_+}{\mu_r+\delta(\mu_r+x_r-y_r)}.   \]

\item[\textnormal{\textit{Effect of changing $S_r$:}}] 
Removing points from $S_r$ can only decrease  $\reppot_r$. Any point $u$ added to $S_r$ comes from $(q-1,p]$, and since these points have only just been passed by the pointer, $R_{r u}=(u,p]$ for such points; as any such point $u$ added to $S_r$ is also in $[x_r,\infty)$, we have $|(y_r,x_r]\cap R_{r u}|=0$ for any $u$ added to $S_r$. Thus, the change of the set $S_r$ does not increase $\reppot_r$ any further.%

\item[\textnormal{\textit{Effect of changing $R_{ru}$:}}] 
Finally and most crucially, we consider the change of $\reppot_r$ resulting from the change of the sets $R_{ru}$. If the current request is not a repeat violation, then just before the pointer reaches position $p$ during the current request, $R_{r p}$ must contain the interval $(1,p)$. Once the pointer reaches $p$, the set $R_{r p}$ becomes empty. In particular, using that $y_r \ge 1$ and given that $x$ is changing only while $x_r<p$, the quantity $|(y_r,x_r]\cap R_{r p}|$ then changes from $(x_r-y_r)_+$ to $0$. %
If $p\in S_r$, this contributes to a decrease of $\reppot_r$. Since the pointer $p$ moves at speed $8$, the contribution of the change of the set $R_{r p}$ to $\reppot_r'$ is then
\begin{align*}
    &-\1_{\{p\in S_r\}}\frac{8(x_r-y_r)_+}{\mu_r+\delta(\mu_r+x_r-y_r)}.
\end{align*}
However, if there \emph{is} a repeat violation, then in the worst case $R_{r p}$ may already be empty so this would not yield any contribution to $\reppot_r$. In this case, by Lemma~\ref{lem:posSeq}, $\epsilon_r$ increases by $1$ due to this request. In continuous time, as the pointer moves for $1/8$ unit of time, this corresponds to increasing $\epsilon_r$ at rate $\epsilon_r'=8$. We claim that regardless of a repeat violation or not, the contribution of the change of $R_{r p}$ to $\reppot_r'$ is at most
\begin{align*}
    &-\1_{\{p\in S_r\}}\frac{8(x_r-y_r)_+}{\mu_r+\delta(\mu_r+x_r-y_r)} + \ell\epsilon_r'.
\end{align*}
In case of no repeat violation this is just our statement above. If there is a repeat violation then since $\mu_r\ge 0$, $\delta=1/\ell$ and $\epsilon_r'=8$ this quantity is $\ge 0$, and since $R_{rp}$ becoming empty cannot increase $\reppot_r$, this is a valid upper bound.

Note that $\reppot_r$ can also increase as points are added to some $R_{r u}$. The only point added to any $R_{r u}$ is the current pointer position $p$. The intersection $(y_r,x_r]\cap R_{r u}$ can be increased by this only if $p\le x_r$ (when $x_r$ is \emph{not} changing). During those times, the integrand increases at most at rate $8/\mu_r$ and only if $p\in(y_r,x_r]$. This can cause $\reppot_r$ to increase at rate at most $\1_{\{y_r<p\}}8$.
\end{description}

\noindent Overall, while $x_r$ is changing, $\reppot_r$ changes at rate
\begin{align*}
    \reppot_r'&\le \1_{\{y_r<p\}} + \1_{\{p\in S_r\}}\frac{2(x_r-y_r)_+}{\mu_r+\delta(\mu_r+x_r-y_r)} -\1_{\{p\in S_r\}}\frac{8(x_r-y_r)_+}{\mu_r+\delta(\mu_r+x_r-y_r)}+ \ell\epsilon_r'\\
    &\le \1_{\{y_r<p\}} + \ell\epsilon_r' - \1_{\{p\in S_r\}}\frac{6(x_r-y_r)_+}{\mu_r+\delta(\mu_r+x_r-y_r)},
\end{align*}
and while $p$ is moving but $x_r$ is not changing, $\reppot_r$ changes at rate
\begin{align*}
    \reppot_r'\le \1_{\{y_r<p\}}8.
\end{align*}

Combining this with our bound on the change of $\reppot_i$ for $i\ne r$, we obtain the bounds on the change of $\reppot$ that are summarized in the following lemma.
\begin{lemma}\label{lem:LambdaIncr}
The overall rate of change of $\reppot$ can be upper bounded as 
\begin{align*}
    \reppot' &\le \1_{\{y_r<p\}}w_r + \ell\epsilon' - \1_{\{p\in S_r\}}\frac{6w_r(x_r-y_r)_+}{\mu_r+\delta(\mu_r+x_r-y_r)} + \sum_{i\ne r}
    \frac{w_i(x_i-y_i)_+\mu_i'}{\mu_i+\delta(\mu_i+x_i-y_i)} \qquad \text{(if $x$ is changing),}\\
    \reppot'&\le \1_{\{y_r<p\}}8w_r \qquad \qquad\qquad \qquad\qquad \qquad \qquad  \qquad \text{(while $p$ is moving but $x$ is unchanged).}
\end{align*}
\end{lemma}
The $\1_{\{y_r<p\}}w_r$ term can be charged to the offline cost, and then note crucially that the third and fourth terms in $\reppot'$ (when $x$ is changing) are the same but with opposite signs (and up to a factor $5/2$) as the third and fourth term in our bound on $\entpot'$. 

\subsection{\texorpdfstring{The scatter potential $\scatpot$ and its rate of change}{The scatter potential and its rate of change}}\label{sec:scatterPot}
It remains to cancel the term $\1_{\{y_r < x_r + \mu_r \;\land\; p\notin S_r\}}2w_r$ in our bound on $\entpot'$ 
. If $S_r$ were the single interval $[x_r,x_r+\mu_r]$, then this term would be $\1_{\{y_r < p\}}2w_r$ and we could charge it to the offline cost.
However, in general $S_r$ can be a union of several intervals of total size $\mu_r$, scattered somewhere in $[x_r,\infty)$ and containing gaps within $[x_r,x_r+\mu_r]$. In this case, such a bound need not hold, and we use the following scatter potential.
\[
  \text{(Scatter potential)} \qquad   \scatpot:=\sum_{i=1}^\ell w_i\left(|S_i\cap[y_i,\infty)|+(x_i-y_i)_+\right).
\]
\begin{lemma}
\label{lem:scatter-change}
The rate of decrease of the scatter potential $\scatpot$ satisfies
\[\scatpot'\le -\1_{\{y_r<x_r+\mu_r\}}w_r + \1_{\{y_r<p\}}2w_r.\]
\end{lemma}
\begin{proof}
For $i\ne r$, the term $(x_i-y_i)_+$ is non-increasing. The term $|S_i\cap[y_i,\infty)|$ can increase for $i\ne r$ only if $x_r>y_r$ (as points are added to $S_i$ at $x_i$). Thus, any possible increase of $|S_i\cap[y_i,\infty)|$ is cancelled by a decrease of $(x_i-y_i)_+$. So $\scatpot'$ is bounded by the rate of change for $i=r$.

A possible increase of the term $(x_r-y_r)_+$ would be cancelled by a decrease of $|S_r\cap[y_r,\infty)|$ caused by removing points from the left of $S_r$ (as $S_r\subset[x_r,\infty)$). The removal of points from the right of $S_r$ at rate $1$ contributes a decrease at rate $w_r$ if $y_r < \sup S_r$. Since $x_r+\mu_r\le \sup S_r$, this contributes at most $-\1_{\{y_r<x_r+\mu_r\}}w_r$ to the change of $\scatpot$. Any other change of $\scatpot$ could only be due to adding points from $(q-1,p]$ to $S_r$ at rate $2$. This can increase $\scatpot$ only if $y_r<p$ and hence contributes $\1_{\{y_r<p\}}2w_r$. Together, this gives the claimed bound.
\end{proof}

\subsection{Putting it all together}
Consider our overall potential $\overpot = 2\entpot+5\reppot+4\scatpot$. Using the bounds from Lemmas \ref{del-phi-paging}--\ref{lem:scatter-change}, we see that while $x$ is changing, this potential is changing at rate
\begin{align*}
    2\entpot'+5\reppot'+4\scatpot'
    \le&-\frac{1}{B}\quad+\quad \1_{\{y_r<p\}}w_r O(\log \ell) \quad-\quad 5\sum_{i\ne r}\frac{w_i(x_i-y_i)_+\mu_i'}{\mu_i+\delta(\mu_i+x_i-y_i)} \\
    &\qquad + \1_{\{p\in S_r\}}\frac{30w_r(x_r-y_r)_+}{\mu_r + \delta(\mu_r+x_r-y_r)} \quad+\quad \1_{\{y_r < x_r + \mu_r \text{ and } p \notin S_r\}}4w_r \\
    &\qquad +\ell\epsilon' \quad -\quad  \1_{\{p\in S_r\}}\frac{30w_r(x_r-y_r)_+}{\mu_r+\delta(\mu_r+x_r-y_r)} \quad+\quad \sum_{i\ne r}
    \frac{5w_i(x_i-y_i)_+\mu_i'}{\mu_i+\delta(\mu_i+x_i-y_i)}\\
    &\qquad-\1_{\{y_r<x_r+\mu_r\}}4w_r\\
    \le& -\frac{1}{B} \quad+\quad \1_{\{y_r<p\}}w_r O(\log \ell) \quad +\quad \ell\epsilon',
\end{align*}
where in the first step the term $\1_{\{y_r<p\}}w_r O(\log \ell)$ absorbs any other terms of the form $\1_{\{y_r<p\}}O(1)w_r$.

\paragraph{The competitive ratio.} 
By the bounds in Section \ref{sec:offlinecost}, when $y$ is changing, the increase in potential is $O(\log \ell)$ times the offline cost, and we can charge additional cost at rate $\1_{y_r<p} 8w_r$ to the offline algorithm while the pointer $p$ is moving.
By Lemma \ref{lem:pseudo}, the online algorithm suffers pseudo-cost at rate $1/B$ while it is moving.
Thus, inequality \eqref{eq:potMainBound} follows from the bound on the increase in potential above. Integrating over time, we get
\begin{align*}
    \on &\le O(\log\ell)\Off + O(\ell)\epsilon+O_{\ell,w,n}(1)\\
    &\le O(\log\ell)(3\OPT+\epsilon) + O(\ell)\epsilon+O_{\ell,w,n}(1)\\
    &\le O(\log\ell+\ell\epsilon/\OPT)\cdot\OPT+O_{\ell,w,n}(1),
\end{align*}
where the second inequality uses Lemma~\ref{lem:optInCi} and $O_{\ell,w,n}(1)$ denotes any constants that may depend on $\ell, n$, the weights, but are independent of the input sequence. We conclude that our algorithm is $O(\log\ell)$-competitive in case of perfect predictions and $O(\log\ell+\ell\epsilon/\OPT)$-competitive in general.

{\small
\bibliographystyle{abbrv}
\bibliography{cache.bib}
}
\appendix

\newpage
\section{A simple analysis of Belady's \FIF algorithm}
\label{sec:belady}
We present a simple potential function argument for Belady's \FIF algorithm for unweighted paging,
that evicts the page whose next request is the farthest in the future.

We first set up some notation. At any time $t$, let $C(t)$ be the set of pages in the cache of \FIF and $C^*(t)$ be those in the cache of some fixed offline optimum solution.
At any time $t$, let us order the pages according to their next request (this order only depends on the request sequence and not on $C(t)$ or $C^*(t)$). This order evolves as follows. At time $t$, the page at position $1$ is requested. It is then reinserted in some position $v$ and the pages in positions $2,\ldots,v$ previously move one position forward to $1,\ldots,v-1$, while the pages in positions $v+1,\ldots,n$ stay unchanged.

Let $\pos(p,t) \in [n]$ denote the position of page $p$  at time $t$ in the order above.
If a page needs to evicted at time $t$, \FIF evicts the page in $C(t)$ with the largest position.
Let $n(s,t) = |\{ p \mid \pos(p,t) \geq s, p \in C(t)\}|$ denote the number of pages in the cache of \FIF with position is at least $s$. Similarly, let $n^*(s,t) = |\{ p \mid \pos(p,t) \geq s, p \in C^*(t)\}|$. Let $e(s,t) = n(s,t) - n^*(s,t)$. We define the following potential function at time $t$
\[ \Phi(t) = \max_s \{n(s, t) - n^*(s, t)\} = \max_s e(s, t).\]
In other words, this is maximum over all suffixes of the ordering, of the difference between pages from $C(t)$ and $C^*(t)$ in that suffix. Note that $n(1,t)=n^*(1,t)=k$, for all $t$ and hence $\Phi(t)\geq 0$. For any function $f:[T] \rightarrow \bbR$, let $ f'(t) := f(t) - f(t-1)$.
\begin{lemma}
\label{lem:belady-pot}
For any time $t$, let $\FIF(t)$ and $\OPT(t)$ denote the total cost incurred up to time $t$ by \FIF and the offline optimum respectively. Then,
\begin{equation}
    \FIF'(t) + \Phi'(t) \leq \OPT'(t). \label{eq:belady-eq}
\end{equation} 
\end{lemma}
\begin{proof}
We show that \eqref{eq:belady-eq} holds in three steps: when offline serves the request at $t$, then when online serves the request, and finally when the ordering of pages changes. Let $p$ be the request at time $t$.

\noindent
{\bf $\OPT$ moves.}  When $\OPT$ serves $p$, if it evicts some page $q$, then $\Phi(t)$ increases by at most $1$, as $n^*(s,t)$ changes by at most $1$ for any $s$. So, $\Phi'(t) \leq 1 = \OPT'(t)$. On the other hand, if $\OPT(t)$ does not evict any page, then $\Phi'(t) = \OPT'(t) = 0$.

\noindent
{\bf $\FIF$ moves.} If $p$ is already in the cache, then $\FIF'(t) = \Phi'(t) = 0$. So, we assume that $p$ (of position 1) is not in \FIF's cache. But as this page is in $\OPT$'s cache (as it has just served the request), we have $e(2, t) = k - (k-1) = 1$, and hence $\Phi(t) \geq 1$. Let $r = \max \{\pos(p,t) \mid p \in C(t)\}$ be the position of the farthest in future page in \FIF's cache. By definition, we have $e(s, t) \leq 0,\ \forall s \geq r+1$. As \FIF evicts page $q$ with $\pos(q,t) = r$, $e(s, t)$ decreases by 1 for all $s \in [2, r]$, and thus $\Phi'(t) = -1$. Thus, $\FIF'(t) +  \Phi'(t) = 0$ as desired.

Finally, consider when $t$ gets incremented. As $p$ is inserted in position $v$: position of $p$ becomes $v$, the positions of pages in positions $2,\ldots,v$ decrease by $1$, and pages with positions $\geq v+1$ stay unchanged. As $p$ lies both in the offline and online cache (as it was just served), for every $s$ both $n(s, t)$ and $n^*(s, t)$ change by the same amount and hence $\Phi'(t) = 0$.
\end{proof}

\begin{theorem}
\FIF is 1-competitive for unweighted paging.
\end{theorem}
\begin{proof}
The optimality of \FIF now follows directly from Lemma~\ref{lem:belady-pot} as $\Phi(t) \geq 0$ for all $t$, and
\[
    \FIF(t) = \sum_{\tau = 1}^t  \FIF'(\tau) \leq \sum_{\tau=1}^t (\OPT'(\tau) - \Phi'(\tau)) = \OPT(t) - \Phi(t) \leq \OPT(t).
\]
\end{proof}

\newpage 
\section{Deterministic algorithm}
\label{sec:det-upper}
We now give a natural extension of the \FIF algorithm to the weighted case that yields an $\ell$-competitive deterministic algorithm for learning-augmented weighted paging with perfect predictions (Theorem~\ref{thm:det}), and an $\ell+2\ell\epsilon/\OPT$-competitive deterministic algorithm in the case of imperfect predictions. Combined with any $k$-competitive online algorithm for weighted paging \cite{ChrobakL91,young,BartalK04} using the method of \cite{FiatRR94,AntoniadisCEPS20} to deterministically combine several online algorithms, this yields Theorem~\ref{thm:det-errors}.

Let $\pos_i(p,t)$ denote the position of page $p$ among all pages of weight $w_i$ when the pages are sorted by the time of their predicted next request, just before the $t$-th request.%

\paragraph{Algorithm.}
For each weight class $i$, we maintain a water-level $\alpha_i(t)$ that is initialized to $w_i$. At any time $t$ when a page eviction is necessary, the page to evict from cache is decided as follows.

Let $J(t)\subseteq[\ell]$ be the set of weight classes from which the algorithm holds at least one page in its cache. %
Let $i' = \argmin_{i\in J(t)} \alpha_i(t)$ be the weight class among them with the least level (ties broken arbitrarily). Then we evict the page from weight class $i'$ with the highest position and set
the levels as 
\begin{align*}
    \alpha_i(t+1) = \begin{cases}
        w_{i'} &i=i'\\
        \alpha_i(t) - \alpha_{i'}(t)\quad& i \in J(t)\setminus\{i'\}\\
        \alpha_i(t)\quad& i \not\in J(t).
    \end{cases}
\end{align*}
In other words, the level of the class $i'$ from which the page is evicted is reset to $w_{i'}$, and the levels for all other classes with at least one page in the cache are decreased by $\alpha_{i'}(t)$.

\paragraph{Potential function analysis.}
For any weight class $i$, let $C_i(t)$ and $C_i^*(t)$ be the set of pages of weight $w_i$ maintained in cache by the online algorithm and some optimal offline algorithm, respectively, just before serving the request for time $t$. Let $n_i(s,t) = |\{p \in C_i(t) \mid \pos_i(p, t) \geq s\}|$ denote the number of pages of weight $w_i$ in the cache of the online algorithm whose position is at least $s$. Similarly, let $n^*_i(s,t) = |\{p \in C^*_i(t) \mid \pos_i(p, t) \geq s\}|$.

Let $e_i(t) = \max_s (n_i(s, t) - n_i^*(s,t))$ denote the ``excess'' for weight class $i$. We note that $e_i(t) \geq 0$, since $s$ can be chosen greater than the maximal position so that $n_i(s, t) = n_i^*(s, t)=0$. For each weight class $i$, we also define a  term $\beta_i(t)$ as follows.
\begin{equation}
    \beta_i(t) = 
    \begin{cases}
        w_i(e_i(t) - 1) + \alpha_i(t) &\text{if $e_i(t) \geq 1$}\\
        0 &\text{otherwise ($e_i(t) = 0$)}.
    \end{cases}
\end{equation}
We define the following potential function at time $t$:
\[\Phi(t) = \sum_{i \in [\ell]} \ell \beta_i(t) - \alpha_i(t).\]
\begin{theorem}
The algorithm is $\ell+2\ell\epsilon/\OPT$-competitive for learning-augmented weighted paging.
\end{theorem}
\begin{proof}
It suffices to show for each time step that
\begin{equation}
    \label{pot:eq:wbel}
\Delta\on + \Delta\Phi \leq \ell\, \Delta\OPT+2\ell\Delta\epsilon,
\end{equation}
where $\Delta\on$ and $\Delta\OPT$ denote the cost incurred in this time step by the online algorithm and the optimum offline algorithm respectively, $\Delta\Phi$ is the associated change in potential, and $\Delta\epsilon$ is the increase of the prediction error $\epsilon$.

Consider any fixed time step $t$ where page $\sigma_t$ is requested, and let $r$ denote the weight class of $\sigma_t$. Note that either the prediction is correct and $\pos_r(\sigma_t,t) = 1$ or otherwise $\Delta\epsilon=w_r$. For ease of analysis, we consider the events at time $t$ in three stages and will show that~\eqref{pot:eq:wbel} holds for each of them: (1) First the offline algorithm serves the page request, then (2) the online algorithm serves the request and a possible increase of $\epsilon$ is charged, and finally (3) the positions of pages in weight class $r$ are updated.
\begin{enumerate}
    \item If $\sigma_t$ is already in the offline cache, then \eqref{pot:eq:wbel} holds trivially. Otherwise, let $j$ be the weight class from which the offline algorithm evicts a page. Then $e_j$ can increase by at most $1$ and no other $e_i$ can increase, so $\Delta\Phi \leq \ell w_j$. As $\Delta\OPT = w_j$, \eqref{pot:eq:wbel} holds.
    \item Now consider the actions of the online algorithm. If $\sigma_t$ is already in the online cache, then \eqref{pot:eq:wbel} holds trivially.
    So suppose the requested page is not in the cache, and let $i'=\arg\min_{i\in J(t)}\alpha_i(t)$ be the class from which the online algorithm evicts a page.
    We consider the following three substeps: (a) level $\alpha_i$ of  each class $i\in J(t)$ is decreased by $\alpha_{i'}(t)$, (b) a page is evicted from class $i'$ and $\alpha_{i'}$ is reset from $0$ to $w_{i'}$, and finally (c) the requested page $\sigma_t$ is fetched into cache. Since $\Delta\on=w_{i'}$, to get \eqref{pot:eq:wbel} it suffices to show that $\Delta\Phi\le \ell\Delta\epsilon$ in steps (a) and (c), and $w_{i'}+\Delta\Phi\le0$ in step (b).
    \begin{enumerate}
    \item If there exists a class $i\in J(t)$ with $e_i(t) \geq 1$: The decrease of $\alpha_i$ by $\alpha_{i'}(t)$ contributes $- \ell \alpha_{i'}(t)$ to $\Delta\Phi$ due to the first term in $\Phi$. The second term can increase by at most $\alpha_{i'}(t)$ for each class, so overall $\Delta\Phi \leq -\ell \alpha_{i'}(t) + \ell \alpha_{i'}(t) = 0 $.
    
    Otherwise, we have $e_i(t)=0$ for each $i\in J(t)$, and clearly $e_i(t)=0$ holds also for $i\not\in J(t)$. In particular, the online and offline algorithm have the same number of pages in cache from each class. Since the offline cache contains page $\sigma_t$, this means that also the online cache must contain some page from class $r$, so $r\in J(t)$. Then $\Delta\Phi\le \ell\alpha_{i'}(t)\le\ell\alpha_r(t)\le\ell w_r$. To conclude this step, it suffices to show that $\Delta\epsilon=w_r$. Suppose not, then $\pos_r(\sigma_t,t)=1$. But then the fact that $\sigma_t$ is in the offline but not the online cache and that they have the same number of pages from class $r$ in their cache would imply that $e_r(t)\ge n_r(2,t)-n_r^*(2,t)=1$, a contradiction.
     
    \item We claim that $\Delta\Phi = -w_{i'}$. As $\alpha_{i'}$ is reset from $0$ to $w_{i'}$, the second term in $\Phi$ contributes $-w_{i'}$ to $\Delta\Phi$. 
If $e_{i'}(t) \geq 1$, then $e_{i'}$ decreases by $1$ upon the eviction because the evicted page has maximum position, and as $\alpha_{i'}$ changes from $0$ to $w_{i'}$, $\beta_{i'}$ stays unchanged. If $e_{i'}(t)=0$, $\beta_{i'}$ does not change anyways. 
\item If $\pos_r(\sigma_t,t)=1$, then fetching $\sigma_t$ to the online cache does not change $e_r$ (as $\sigma_t$ was already in the offline cache) and therefore $\Delta\Phi=0$. Otherwise, we have $\Delta\epsilon=w_r$ and $e_r$ increases by at most $1$, so $\Delta\Phi\le \ell w_r=\ell\Delta\epsilon$.
    \end{enumerate}
    \item Finally, when page $\sigma_t$ is re-inserted in some position of the predicted order, this does not affect $e_r$ because page $\sigma_t$ is in both the online and offline cache. Therefore, $\Delta\Phi =0$ and \eqref{pot:eq:wbel} holds.
\end{enumerate}
The theorem now follows by summing up \eqref{pot:eq:wbel} over all time steps. 
\end{proof}
\newpage
\section{ Missing proofs from Section \ref{sec:branking} }
\label{app:missing-ranking-proofs}

\begin{lemma}[Consistency, repeated Lemma~\ref{lem:consistency}]%
Let $C^1_0, C^2_0,\ldots$ be any initial configurations, satisfying $C^\position_0\subset C^{\position+1}_0$ for all $\position$.
Then for any sequences $\sigma, \tau$, for all times $t$ and all $\position \geq 0$, we have  $C^{\position}_{t}(\sigma,\tau) \subset C^{\position+1}_{t}(\sigma,\tau)$.
\end{lemma}

\begin{proof}
We use induction on time $t$. The base case for $t=0$ holds by assumption. 
For the induction step, let $p$  be the page requested at time $t+1$, and let $a$ denote the smallest index such that $p \in C^{a}_t$.\footnote{Note that such an index $a$ must exist since we are guaranteed that $r \in C^n_t$ where $n$ is the total number of pages.} By the inductive hypothesis, as $C^\position_t \subset C^{\position+1}_t$
for all $\position$, we have that $p$ lies in $C^{\position}_t$ for all $\position\geq a$, and none of these caches incur a page fault. So $C^{\position}_{t+1} =C^{\position}_t$ and the property $C^{\position}_{t+1} \subset C^{\position+1}_{t+1}$ for all $\position \geq a$ is maintained.

For $\position<a$, each cache $C^\position_t$ evicts its page with the farthest predicted re-arrival time and fetches $p$. 
Let us consider this in two steps.
First, adding $p$ to each $C^\position_t$ for $\position<a$ maintains the property that $C^{\position}_{t+1} \subset C^{\position+1}_{t+1}$ for all $\position$, since $p\in C^a_{t+1}$.
Let us now consider the eviction step.
Fix some $\position< a$, and suppose $C^{\position+1}_t$ evicts $q$. If $q \in C^{\position}_t$, then as $C^{\position}_t\subset C^{\position+1}_t$ by the inductive hypothesis, $q$ is also the page with the farthest predicted re-arrival time in $C^{\position}_t$ and hence evicted from $C^{\position}_t$. 
Otherwise $q \notin C^{\position}_t$ and some other page $q'$ is evicted from $C^{\position}_t$. In either case, $C^{\position}_{t+1} \subset C^{\position+1}_{t+1}$ is maintained for all $\position < a$.
\end{proof}

\begin{lemma}[repeated Lemma~\ref{lem:optInCi}]
Let $A$ be an arbitrary (offline) weighted paging algorithm, and let $x_{i}(t)$ denote the cache space used by class $i$ at time $t$ under $A$. For any arbitrary $\sigma,\tau$, let $A^*$ be the trustful algorithm with configuration $\bigcup_{i=1}^\ell C_{i,t}^{x_{i}(t)}(\sigma,\tau)$ at any time $t$. Then,
	$	\cost_{A^*}(\sigma,\tau)\le 3\cdot\cost_{A}(\sigma)+\epsilon(\sigma, \tau)+O(1)$.
\end{lemma}
\begin{proof}
We will use a potential function for analysis. For any weight class $i$ at any time $t$, we order all pages of class $i$ in increasing order of the predicted arrival time of their next requests (breaking ties in the same way as \Blind). We call the position of pages in this ordering the \emph{predicted position}. Note that this predicted position of pages within a weight class is very different from the \emph{rank} of pages defined in Section~\ref{sec:ranks}.
We observe the following property: at each time step $t$, when a page $\sigma_t$ from some weight class $i$ is requested, either $\sigma_t$ has predicted position 1 in its class or the request contributes to the prediction error $\epsilon_i$. 

For any integer $s$, let $n_i(s)$ denote the total number of pages in the cache of algorithm $A$ with predicted position at least $s$. Let $n^*_i(s)$ denote the respective quantity for algorithm $A^*$. We note that these quantities vary with time $t$, but we suppress the dependence in the notation for brevity. Let $\Phi_i := \max_s n^*_i(s) - n_i(s)$ and consider the potential function
\[\Phi := 2 \sum_{i=1}^{\ell} w_i \Phi_i.\]
We consider the setting where algorithms pay cost $w_i$ whenever they evict \emph{or} fetch a page of weight class $i$. As this doubles the cost of any algorithm compared to the original setting where algorithms only pay for evictions, up to an additive constant, it suffices to show that for each request,
\begin{align}\label{eq:potCanonical}
\Delta \cost_{A^*}+\Delta\Phi\le 3 \Delta\cost_{A}+2\Delta\epsilon,
\end{align}
where $\Delta \cost_{A^*}$ and $\Delta\cost_{A}$ are the costs incurred for this request, $\Delta\Phi$ is the change in potential and $\Delta\epsilon$ is the increase of $\epsilon$.

For any request to some page $p$ from weight class $i$, we break the analysis into three steps: (1) First $A$ serves the request and $A^*$ updates its cache accordingly with respect to the \emph{old} ranks of each weight class. (2) Then $A^*$ updates its cache content to reflect the new ranks of weight class $i$. (3) Finally page $p$ might move to a later position in the predicted order for class $i$. In each step, we will show that inequality~\eqref{eq:potCanonical} is satisfied.

In step (1), suppose $A$ evicts page $q$ from some weight class $j$. In this case, $\Delta\cost_A=w_i+w_j$ and $\Delta\cost_{A^*}\le w_i+w_j$. Moreover, both $\Phi_i$ and $\Phi_j$ increase by at most 1 and hence $\Delta \Phi \leq 2 (w_i + w_j)$ and thus the inequality is maintained.

In step (2), after page $p$ is requested the ranks of pages of weight class $i$ change according to Lemma \ref{lem:orderingUpdate} (see Figure~\ref{fig:BRUpdate}). In particular, $p$ moves to rank 1 and, if $p$ is not in cache yet, $p$ is fetched and the page in cache from class $i$ with the highest predicted position gets evicted. Let $q$ be the evicted page. %
When $A^*$ fetches page $p$ and evicts $q$, it incurs a cost of $\Delta \cost_{A^*} = 2w_i$. We analyze the change in potential $\Delta \Phi$ and error $\Delta \epsilon$ due to fetching of page $p$ and evicting $q$ separately.

Let $s$ be such that $\Phi_i = n^*_i(s) - n_i(s)$. We observe that the predicted position of page $q$ must be at least $s$ (since otherwise, we would have $n^*_i(s) = 0$ and hence $\Phi_i \leq 0$, but $\Phi_i \geq n^*_i(1) - n_i(1) = 0$). Thus, evicting $q$ decreases $\Phi_i$ by 1 and we have $\Delta \Phi = -2 w_i$ and inequality \eqref{eq:potCanonical} is maintained. To account for the change in potential due to fetching page $p$, we consider separately the cases that $p$ has (old) predicted position 1 in class $i$ or not. In the former case, we have $s \geq 2$ and hence the potential does not change. Otherwise, $\Phi_i$ increases by at most 1, so $\Delta \Phi \leq 2 w_i$. However, in this case, the prediction error $\epsilon_i$ also increments and we have $\Delta \eps = w_i$ and inequality \eqref{eq:potCanonical} is maintained.

Finally in step (3), when page $p$ is re-inserted in some position of the predicted order, the potential is not affected since now $p$ is present in the cache of both algorithms.
\end{proof}

\begin{Remark}
A slight modification of the proof of this lemma yields a much simpler proof of the result from~\cite{Wei20} that in unweighted paging, $\Blind(k)$ has competitive ratio at most $1+\epsilon/\OPT$: In this case, we have a single weight class and $x_1(t)=k$ remains fixed. Therefore, $A^*=\Blind(k)$ does nothing in step (1), and we can avoid losing a factor $3$ by considering the setting where algorithms are charged only for evictions (not for fetching) and omitting the factor $2$ in the definition of $\Phi$.
The proof of this result in~\cite{Wei20} uses a case analysis involving eleven cases.
\end{Remark}

\begin{lemma}[Repeat property, repeated Lemma~\ref{lem:repeat}]
Let $\epsilon=0$ and let $i$ be a weight class. A rank sequence corresponds to a request sequence of pages of class $i$ if and only if it has the following \emph{repeat property}: for any $h$, between any two requests to the same rank $h$, every rank $2,\dots,h-1$ must be requested at least once.
\end{lemma}

\begin{proof}
Lemma~\ref{lem:posSeq} shows that if $\eps=0$, any rank sequence must satisfy the repeat property.

Conversely, we show that for any sequence $h_1,h_2,\dots$ satisfying the repeat property, there exists a corresponding paging request sequence $r_1,r_2,\dots$.
Let $n:=\max_t h_t$ be the number of distinct pages. We construct the request sequence online by specifying, whenever a page is requested, the time when the same page will be requested next. 

We will maintain the invariant that for each $t$ and integer $m=1,\dots,n$, after the request at time $t$, the page with rank $m$ has next-request time given as follows.
\begin{align*}
    &\inf\{t' > t\colon h_{t'}=m\}&&\qquad\text{if $m\ge 2$ or $m=h_{t+1}$,}\\
    &\inf\{t' > t+1\colon h_{t'}=h_{t+1}\}&&\qquad\text{if $m=1$ and $m\ne h_{t+1}$.}
\end{align*}
The invariant ensures that the next request will be to the page in position $h_{t+1}$, as required. It also implies that different pages have different next-request times. It remains to show that we can maintain this invariant over time. We can satisfy the invariant initially by defining the first-request times of the pages according to the condition of the invariant for $t=0$.

Suppose the invariant holds for some $t$. By the invariant, the page $r_{t+1}$ requested at time $t+1$ is the one with rank $h_{t+1}$. We define the next request time to page $r_{t+1}$ to be $t+2$ if $h_{t+2}=1$ and $\inf\{t'>t+2\colon h_{t'}=h_{t+2}\}$ if $h_{t+2}\ge 2$. Since $r_{t+1}$ will receive new rank $1$, we see that the condition for $m=1$ of the invariant is satisfied for the next time step. If $h_{t+1}=1$, then the ranks remain unchanged and the invariant continues to be satisfied. So suppose $h_{t+1}\ge 2$. By assumption on the sequence $h_1,h_2,\dots$, for each $m=2,\dots,h_{t+1}-1$ we have
\begin{align*}
\inf\{t' > t+1\colon h_{t'}=h_{t+1}\} \ge \inf\{t' > t+1\colon h_{t'}=m\} = \inf\{t' > t\colon h_{t'}=m\},
\end{align*}
with the inequality being strict unless both sides are $\infty$. Thus, by the invariant, the page previously in position $1$ is the one with the farthest next-request time among the pages in positions $1,2,\dots,h_{t+1}-1$.
so the new ranks are the same as the old ones except that the pages with ranks $1$ and $h_{t+1}$ swap, and it is directly verified that the invariant is again satisfied at the next time step.
\end{proof}

\end{document}